\documentclass{article}

\usepackage[letterpaper,left=1in, right=1in, top=1in, bottom=1in]{geometry}

\usepackage{graphicx} % Required for inserting images

\usepackage{algorithm}
\usepackage{algorithmic}

\usepackage{amssymb,amsmath,amsfonts,enumitem,amsthm,physics,mathtools,wrapfig}
\usepackage{thm-restate}
\usepackage{thmtools}
\usepackage{xcolor}
\usepackage{hyperref}
\usepackage{newfloat}
\usepackage{listings}
\usepackage{wrapfig}
\usepackage{subcaption}

\usepackage[utf8]{inputenc} % allow utf-8 input
\usepackage[T1]{fontenc}    % use 8-bit T1 fonts
\usepackage{url}            % simple URL typesetting
\usepackage{booktabs}       % professional-quality tables
\usepackage{nicefrac}       % compact symbols for 1/2, etc.
\usepackage{microtype}      % microtypography

\newcommand{\E}{\mathbb{E}}

\newcommand{\N}{\mathbb{N}}

\newcommand{\plog}{\text{polylog}}

\newcommand{\D}{\mathcal{D}}

\newcommand{\Ess}{\mathbf{S}}

\newcommand{\ahat}{{\widehat{a}}}
\newcommand{\atil}{{\widetilde{a}}}

\newcommand\degr[1]{\text{deg(#1)}}

\newcommand{\ddp}{{\delta_{\operatorname{p}}}}
\newcommand{\eps}{{\epsilon}}
\newcommand{\kcl}{{k\mathbb{C}}}

\newcommand{\fkcl}{{f_\kcl}}

\newcommand{\algK}{\textsc{ERS}}

\newcommand{\Xhat}{\widehat{X}}
\newcommand{\Khat}{\widehat{K}}
\newcommand{\Xtil}{\widetilde{X}}
\newcommand{\Itil}{\widetilde{I}}
\newcommand{\LStil}{\widetilde{LS}_{\kcl}}
\newcommand{\Ktil}{\widetilde{K}}

\newcommand{\Kone}{K^{-1}}
\newcommand{\Konetil}{\widetilde{K^{-1}}}
\newcommand{\Konehat}{\widehat{K^{-1}}}

\newcommand{\Sample}{\textsc{Sample}}
\newcommand{\CliqueAlg}{\textsc{Unbounded.PairWise.Clique.Approx}}
\newcommand{\NearCliqueAlg}{\textsc{Unbounded.PairWise.NearClique.Approx}}
\newcommand{\Turan}{\text{Tur\'anShadow}}

\newtheorem{theorem}{Theorem}[section]
\newtheorem{lemma}[theorem]{Lemma}

\newtheorem{corollary}[theorem]{Corollary}
\newtheorem{definition}[theorem]{Definition}

\title{Scalable K-clique Estimation with Differential Privacy}
\author{Dung Nguyen\thanks{Department of Computer Science, Haverford College, PA, USA; Correspondence to: dnguyen1@haverford.edu}\;, Ritwick Mishra\thanks{Department of Computer Science and Biocomplexity Institute, University of Virginia, VA, USA}\;, Anil Vullikanti\footnotemark[2]}

\date{}

\begin{document}

\maketitle
% \keywords{Subgraph Counting, Differential Privacy}

\begin{abstract}
Counts of $k$-cliques are commonly used metrics in subgraph mining.
Since graphs often have sensitive data, there also has been a lot of work on $k$-clique counts with differential privacy.
However, these metrics have very high global sensitivity, and so more sophisticated techniques have been developed for counting $k$-cliques with privacy.
Smooth sensitivity and
ladder functions were developed for reducing the noise magnitude for private estimates of these metrics.
However, these are computationally very inefficient to estimate.
No polynomial time algorithms are known for smooth sensitivity of $k$-cliques for $k>3$, while the time complexity of ladder functions is lower bounded by the time for exact counts, which does not scale very well.

In this paper, we develop a new highly scalable algorithm for estimating $k$-clique counts with differential privacy.
Our algorithm adapts the ladder function to serve as a smooth upper bound on its local sensitivity, and utilizes the approximation sensitivity framework to calibrate noise with magnitude proportional to an approximation of the bound.
This gives us a significant improvement in the running time.
Experiments show that our method is several orders of magnitude faster than the ladder function based estimates of $k$-clique counts, while the accuracy is similar.
Our algorithm is the first to scale to graphs with millions of edges, and for larger $k$, for which the ladder function algorithm doesn't complete. 

%%%%%%%
% and has been extensively studied.
% They are important from a theoretical perspective, since most of them are hard problems, and in applications in various fields.
% Counting $k$-cliques in a graph is a $W[1]$-hard problem with respect to $k$, thus, it has motivated a long line of studies to approximate it efficiently.
%  Unlike some other small subgraph counting problems, the smooth sensitivity, an alternative to the global sensitivity, of $k$-clique count is intractable in its current form. 
% We instead adapt its ladder function to serve as a smooth upper bound on its local sensitivity, and utilize the approximation sensitivity framework to calibrate noise with magnitude proportional to some approximation of the bound.
%  The approximation of the bound can be efficiently computed, leading to an efficient private algorithm for counting $k$-cliques.
%  By this technique, we obtained a polynomial-time algorithm to estimate the smooth upper bound on the local sensitivity.
% Experiments show that our method is several orders of magnitude faster than previous methods of estimating $k$-clique counts with differential privacy.
\end{abstract}

%%% Local Variables:
%%% mode: latex
%%% TeX-master: "main"
%%% End:

\section{Introduction}

Statistics of different kinds of subgraphs, such as triangles and $k$-cliques, more generally, have been studied extensively in graph mining and network science, e.g., \cite{pagh2012colorful, tsourakakis2011counting, kolountzakis2012efficient, jain2017fast}; 
these have been found useful in a number of applications, such as characterizing and comparing networks, e.g., \cite{bonato2024clique} and identifying cohesive structures, e.g., \cite{pan2023simplifying} (see \cite{ribeiro2021survey} for a survey).
Highly efficient algorithms have been designed for counting $k$-cliques, which scale easily to massive networks, e.g., \cite{jain2017fast}.
Many networks involve sensitive components, and there has been a lot of work on such analyses with privacy protections.
While many privacy models have been proposed,  Differential Privacy (DP) \cite{dwork:fttcs14} has become a standard, since it does not make any assumptions about the adversaries.
Two privacy models have been considered in most prior work on DP for graphs--- edge and node DP. 
Private algorithms have been developed for many graph problems, such as counting small subgraphs, including $k$-cliques, e.g.,~\cite{Kasiviswanathan:2013:AGN:2450206.2450232,zhang:sigmod15, karwa2014private, nissim2007smooth}.

It is challenging to develop accurate, efficient and private algorithms for graph problems such as counting $k$-cliques.
These problems  have very high \emph{global sensitivity}, the maximum change in the output function due to the change in one edge/vertex (see Section~\ref{sec:prelims} for definitions), and DP algorithms such as the Laplacian mechanism, which add noise based on global sensitivity can become very inaccurate.
For instance, under edge DP, the global sensitivity for counts of $k$-cliques is $\Theta(n^{k-2})$, or more specifically $GS_{\kcl} = \binom{n-2}{k-2}$.
The number of 5-cliques in the com-orkut network (with $\sim 10^6$ nodes) is $\sim 10^{10}$ \cite{jain2017fast};
this would make a private estimate of the number of 5-cliques in this network with noise scaled by global sensitivity $\sim 10^{17}$ completely noisy and useless.
% The problem is that global sensitivity considers the \emph{worst case} change in the statistic over all graphs.
A number of DP techniques have been developed using alternatives to global sensitivity, e.g., smooth, restricted and multi-level local sensitivity, and ladder functions~\cite{zhang:sigmod15, karwa2014private, blocki:itcs13, nissim2007smooth}---these add significantly lower level of noise than mechanisms based on global sensitivity, and can lead to significantly better accuracy.
However, these techniques are computationally very expensive.
While smooth sensitivity for triangle counting with edge DP can be computed in polynomial time \cite{karwa:tds13}, no efficient algorithms are known for computing the smooth sensitivity for counts of $k$-cliques for $k>3$ (though no formal NP-hardness results are known).
The complexity of ladder functions for counting $k$-cliques is as high as exact counting (non-private) \cite{zhang:sigmod15}.

There has been a lot of work in graph mining for fast computation of $k$-clique count and other graph statistics, e.g.,  \cite{jain2017fast, eden2018approximating}.
These works have used careful sampling techniques for approximate counts, which are much faster than the naive algorithms, e.g., $\Omega(n^k)$ for $k$-clique counts.
Unfortunately, approximations to smooth sensitivity or ladder functions do not immediately translate to private counts, as shown in \cite{nguyen2023faster, blocki22}.
Nguyen et al. \cite{nguyen2023faster} designed a specific form of approximate smooth sensitivity, and showed that it can be computed efficiently for counts of triangles, leading to a faster private algorithm, with a $1000$-fold speedup over the smooth sensitivity technique \cite{karwa:tds13, nissim2007smooth}, without much impact on accuracy.
However, it is open whether the approximate smooth sensitivity approach of \cite{nguyen2023faster} for triangle counts can be extended extended for $k$-clique counts, $k>3$.
In particular, the privacy analysis of the ladder function \cite{zhang:sigmod15} doesn't work if it is computed with an approximation. 
\smallskip

\begin{table*}[h]
\centering
\footnotesize
\caption{
\small
Total running time of private and non-private $k$-clique counting
methods;
see Table \ref{tab:notation} for summary of notation.
The running times listed for all algorithms other than \textsc{FastApproxCliqueDP} is only for the noise computation;
for \textsc{FastApproxCliqueDP} the bound is the complete time, including the non-private count.
% $\D(G)$ is the degeneracy of $G$ (Section~\ref{sec:turan-shadow});
% $\ahat = \max_{u,v} a_{uv}$;
% $\Khat$ and $\Konehat$ are the maxima, over pairs $(u,v)$,
% of $k$-cliques containing the edge $(u,v)$ and of near-$k$-cliques missing
% the non-edge $(u,v)$, respectively;
% $w(\Ess_k)$ is the Tur\'an-shadow weight (Section~\ref{sec:turan-shadow});
$\|W\|_1$ is as in Definition~\ref{def:W}.
Polylogarithmic factors are suppressed in $\tilde{O}$.
% $T_{\text{exact}}(G, k)$ denotes the running time of an exact
% non-private $k$-clique algorithm.
}
\label{tab:runtime-comparison}
\renewcommand{\arraystretch}{1.6}
\begin{tabular}{@{}lll@{}}
\toprule
Method & Privacy & Total running time \\
\midrule
Approximate, no privacy
$T_{\textsc{Eden}}(G, k) $  \cite{eden2018approximating}
  & none
  & 
$O\!\left(
  \left(
    \frac{n}{f_\kcl(G)^{1/k}} \;+\; \frac{m^{k/2}}{f_\kcl(G)}
  \right)
  \cdot \plog(n,\, 1/\alpha,\, k,\, 1/\delta)
\right)$ \\
Ladder function, exact
\cite{zhang:sigmod15} 
  & $\eps$-DP
  & $O\!\bigl(n \cdot T_{\text{exact}}(G,\,k)\bigr)$ \\
\textsc{FastCliqueDP}, black-box (Section \ref{sec:blackbox})
  & $(\eps,\delta)$-DP
  & $\tilde{O}\!\left(
      \displaystyle\sum_{u,v \in V(G)} T_{\textsc{Eden}}\bigl(G(A_{uv}),\,k-2\bigr)
      \;+\; m + n + \tfrac{\|W\|_1}{\ahat^{2}}\right)$ \\
\textsc{FastCliqueDP}, direct (Section \ref{sec:direct})
  & $(\eps,\delta)$-DP
  & $\tilde{O}\!\left(
      \tfrac{w(\Ess_{k})}{\eps^{2}\,\Khat}
      \;+\; \tfrac{w(\Ess_{k-1})}{\eps^{2}\,\Konehat}
      \;+\; n\D(G)^{k} + m + n + \tfrac{\|W\|_1}{\ahat^{2}}\right)$ \\
\textsc{FastApproxCliqueDP}, black-box & $(\eps,\delta)$-DP & Same as \textsc{FastCliqueDP}, black-box, but for complete computation\\
\textsc{FastApproxCliqueDP}, direct & $(\eps,\delta)$-DP & Same as \textsc{FastCliqueDP}, direct, but for complete computation\\
\bottomrule
\end{tabular}
\end{table*}

\noindent
\textbf{Our contributions.}\\
(1) We develop a new algorithm (\textsc{FastCliqueDP}) for private estimation of $k$-clique counts, which \emph{combines approximations to three different concepts: the ladder function, extending \cite{zhang:sigmod15}, approximate smooth sensitivity \cite{nguyen2023faster}, and non-private clique counts \cite{jain2017fast}}.
We prove rigorous bounds on the approximation to smooth sensitivity (Theorem \ref{theorem:tildeS}) and the running time (Theorem~\ref{theorem:total-time-direct} and~\ref{theorem:total-time-blackbox}) for our algorithm.
The running time of \textsc{FastCliqueDP} is orders of magnitude better than the one based on ladder functions \cite{zhang:sigmod15} (Table \ref{tab:runtime-comparison}), with similar accuracy, which is  significantly better than using global sensitivity.

Transforming the ladder function to handle approximation is technically challenging.
We modified the ladder function by adding specific constraints to create a smooth upper bound on the local sensitivity of the $k$-clique count.
However, the smooth upper bound still includes the count of smaller-sized cliques.
We also adapt a fast algorithm to approximate the $(k-2)$-clique statistic~\cite{eden2018approximating}.
The original algorithm, although it achieves any-constant approximation, has a constant failure probability, which is insufficient for privacy.
Instead, we use a median estimation, to get an inverse polynomial failure probability needed by the framework.
% We therefore invoked multiple instances of it, obtaining a median estimation, to get an inverse polynomial failure probability needed by the framework.
As a result, we obtained a polynomial-time algorithm to estimate the smooth upper bound on the local sensitivity, which we can extend our framework to utilize it to calibrate the noise needed for privacy.

\noindent
(2) All prior work on subgraph counting  assumes the non-private count is exact, to which noise is added.
As a result, the total running time (exact count + noise estimation) is at least as much as the time needed for exact counts
(except for \cite{nguyen2023faster}, who show an expression for triangle counting with approximate counts using smooth sensitivity).
We design \textsc{FastApproxCliqueDP}, which uses an approximate $k$-clique count, and is thus significantly faster than the other methods (Table \ref{tab:runtime-comparison}).

\noindent
(3) We evaluate \textsc{FastCliqueDP} on a diverse set of datasets, and show that it scales to very large networks.
\textsc{FastCliqueDP} has significantly better time complexity over ladder functions, while maintaining high accuracy, within 1\% relative error.
The ladder function implementation does not complete in 48 hours for $k=6$ on a network with little over million edges.
The networks we use here are significantly larger than those used in all prior works on private graph analysis--- the largest network we consider has over 16 million edges, which is 2 orders of magnitude larger than those in \cite{zhang:sigmod15} and an order of magnitude larger than those in \cite{nguyen2023faster}.

Our results suggest that careful approximations are the only approach for scaling  private graph analyses to massive networks. 
However, we find that \textsc{FastApproxCliqueDP} sometimes has low accuracy.
Therefore, a better kind of approximation has to be figured out for private estimates.

%\dungnote{Differential privacy, network science, and graph problems with privacy.}

% \dungnote{Subgraph count definition, with privacy. Issues with global. Alternatives have been proposed, but not efficient.}

% \dungnote{Some efforts in fast subgraph counts with privacy: Blocki et al., Nguyen et al.}

% \dungnote{K-clique count, without and with privacy}

% \dungnote{Technical highlights. Our approaches to k-clique counts, using ladder function, transforming to smooth upper bound of local sensitivity, fast estimation.}

% \dungnote{THeoretical bounds, empirical speedup, etc.}

%%% Local Variables:
%%% mode: latex
%%% TeX-master: "main"
%%% End:

\section{Related work}
\label{sec:related}

For brevity, we focus only on related work on private subgraph counting.
The straightforward method to release private subgraph count (of some query $f(G), G=(E, V)$) is using the Laplace mechanism with Global sensitivity ($GS_f$), defined as $GS_f = \max_{G\sim G'}|f(G) - f(G')|$ for any pair of neighbor graphs $G\sim G'$.
We note that there are several definitions of neighboring, which leads to different privacy models on graphs.
However, Global sensitivity is high for many subgraph counting problems.
Hence, adding noise based on $GS_f$ leads to low accuracy, while noise based on just $LS_f$, the local sensitivity, does not guarantee privacy~\cite{Vadhan2017}.
\cite{nissim2007smooth} develops the notion of smooth sensitivity $S^*_{f,\,\beta}(G) = \max_{G'}\Big(LS_f(G')\cdot e^{-\beta \cdot d(G, G')}\Big)$, where $d(G, G')$ is the swap distance between graphs $G$ and $G'$, and show that adding noise based on $S^*_{f,\,\beta}(G)$ is private.
However, smooth sensitivity becomes computationally much more challenging, since its definition involves considering local sensitivity at distance $t$ for all $t$, which doesn't seem like a polynomial time computation.
\cite{nissim2007smooth} develops polynomial time algorithms for computing the smooth sensitivity exactly for counting \#triangles, which was improved slightly to $\min\{md_{max}, M(n)\}$~\cite{karwa2014private}, where $M(n)$ denotes the time complexity of matrix multiplication of two $n\times n$ matrices.
Polynomial time smooth sensitivity bounds were also shown for a small number of other subgraphs, but no polynomial time algorithms are known beyond that.
Interestingly, no hardness bounds are known for smooth sensitivity; the existing hardness bounds, e.g.,~\cite{karwa2014private, zhang:sigmod15} only give hardness for computing local sensitivity at distance $t$.
In order to handle other subgraph counts,
\cite{karwa2014private} develops a different technique involving local sensitivity of local sensitivity (motivated by the propose-test-release technique~\cite{dwork2014algorithmic}), and use it for private counts of $k$-cliques.
\cite{zhang:sigmod15} develops the technique of \emph{ladder functions} to handle other kinds of subgraphs, and use it to privately count \#$k$-cliques in the graph in time $O(nT(n))$, where $T(n)$ is the time needed to count \#$k$-cliques non-privately.
A few other techniques, such as inverse sensitivity~\cite{asi:neurips20} and propose-test-release~\cite{dwork2014algorithmic} are known.
However, private algorithms for counting most subgraphs including paths and trees are not known. 
The recent work of~\cite{blocki2022make} is related, but only considers approximation in the queries, but not in sensitivity.
As a result, our methods give significantly higher efficiency.
We also note that~\cite{blocki2022make} develop a black-box approach to make certain approximation algorithms differentially private. However, their work requires the function being computed to have ``small'' global sensitivity.

%%% Local Variables:
%%% mode: latex
%%% TeX-master: "main"
%%% End:

\section{Preliminaries}
\label{sec:prelims}

% \begin{wraptable}[19]{r}{3.0in}
% % \begin{footnotesize}
% \begin{scriptsize}
\begin{table}
  \centering
\begin{tabular}{|p{0.7in}|p{2.2in}|}
\hline
\textbf{Notation} & \textbf{Description}\\
\hline
$G(S)$ & Subgraph of $G$ induced by the set $S$\\
$f_H(G)$ &  \#non-induced embeddings of subgraph $H$ in $G$\\
$G^n$ & Set of graphs with $n$ nodes\\
$A_{uv}$ & Common neighbors of $u, v$\\
$a_{uv}$ & $|A_{uv}|$\\
$\ahat$ & $\max_{u,v} a_{uv}$\\
$\atil$ & approximation for $\ahat$ \\
$\D(G)$ & Degeneracy of $G$ (Section~\ref{sec:turan-shadow})\\
$\Khat$ & $\max_{u,v}$ \#$k$-cliques containing $(u,v)$\\
$\Konehat$ & $\max_{u,v}$ \#near-$k$-cliques without $(u,v)$\\
$LS_f(G)$ & Local sensitivity of $f$\\
$GS_f$ & Global sensitivity of $f$\\
$S^{\beta}_f(G)$ & $\beta$-smooth upper bound on $LS_f(G)$\\
$S^{*}_{f,\beta}(G)$ & Smooth sensitivity\\
$\tilde{S}^{\beta}_\kcl(G)$ & Approximate smooth sensitivity\\
$I_t(G)$ & Ladder function of $\fkcl$ (Lemma~\ref{lemma:ladder})\\
$\widetilde{I_t}(G)$ & Approximation to $I_t(G)$\\
$\LStil(G)$ & Approximation to $LS_\kcl(G)$\\
$w(\Ess_k)$ & Tur\'an-shadow weight (Section~\ref{sec:turan-shadow})\\
$T_{\text{exact}}(G, k)$ & Running time of 
non-private $k$-clique algorithm\\
\hline
\end{tabular}
%\end{scriptsize}
\caption{
\small
Notation used in the paper}
\label{tab:notation}
\end{table}

We define some of the notions and definitions used in the rest of the paper; some of the notation is summarized in  Table \ref{tab:notation}. 
A non-induced embedding of a graph $H=(V_H, E_H)$ into a graph $G=(V, E)$ is a mapping $\phi:V_H\rightarrow V$ such that $(\phi(u), \phi(v))\in E$ whenever $(u, v)\in E_H$.
We use $f_H(G)$ to denote the number of non-induced embeddings of a $H$ in $G$. We drop the subscript $H$ when it is clear from the context.
We use $G^n$ to denote the set of all graphs with $|V(G)| = n$.
Let $\fkcl(G)$ be the number of $k$-cliques in a graph $G$.
We say that a function $\tilde{f}$ is a $(\eta,\delta)$-approximation to $f$ if $\tilde{f}(G) \in (1\pm \eta)f(G)$ with probability at least $1-\delta$.
We say that a function $\tilde{f}$ is a $(\gamma,\delta)$-upper approximation to $f$ if $f(G) \leq \tilde{f}(G) \leq e^{\gamma}f(G)$ with probability at least $1-\delta$.
Let $A_{uv}$ denote the common neighbors of $u$ and $v$.
We use $G(S)$ to denote the subgraph of $G$ induced by the set of vertices $S$.

%Figure~\ref{fig:subgraphs} shows several patterns of $H$ whose counts are popularly studied in graph mining, including stars ($2$-star (wedge), $3$-star), triangles (triangle, $2$-triangle), cliques ($3$-clique (triangle), $4$-clique).

%\dungnote{DP models on graphs. PRoblem statements.}

\textbf{Differential privacy} guarantees the output of a randomized mechanism $M$ is similar on similar input.
In graphs, there are several privacy models.
In this work, we consider the edge-privacy model, in which two graphs $G$ and $G'$ are neighbors, denoted $G\sim G'$, if $V(G) = V(G')$ and $E(G) = E(G') \pm \{e\}$, i.e., $G$ and $G'$ differ by exactly one edge $e$.
A randomized algorithm (or mechanism) $M$ is $(\epsilon,\delta)$-differentially private in the edge-privacy model (edge-DP for short) if $\Pr[M(G)\in S] \leq e^{\epsilon}\Pr[M(G')\in S] + \delta$ for any $S\subseteq \text{Range}(M)$ \cite{dwork:fttcs14}.

\textbf{Problem Statement (Private $k$-clique counting).} 
Given a graph $G$, parameters $k, \epsilon, \delta$, design an $(\epsilon,\delta)$-edge DP, polynomial time mechanism $M_\fkcl: G^n \rightarrow \mathbb{N}$ that minimizes $|M_\fkcl(G) - \fkcl(G)|$.

%\dungnote{Smooth sensitivity. Ladder function.}

The local sensitivity of function $f(\cdot)$ is defined as $LS_f(G) = \max_{G':G\sim G'}|f(G) - f(G')|$.
Noise calibrated to $LS_f$ doesn't always ensure DP.
One standard method to design DP mechanisms is the Laplace mechanism, adding noise calibrated by the Global sensitivity $GS_f = \max_{G}LS_f(G)$ to the output of $f$, i.e., $M_f(G) = f(G) + Lap(GS_f/\epsilon)$, which is $\epsilon$-DP \cite{dwork:fttcs14}.
$GS_{\fkcl}$ can be very large (as large as $\binom{n-2}{k-2}$ for graphs with $n$ vertices), and so the global sensitivity based mechanism is highly inaccurate, in general. 
The smooth sensitivity~\cite{nissim2007smooth} is an alternative technique, and can lead to much higher accuracy.
Let $LS^{(t)}_f(G) = \max_{G': d(G,G')\leq t}LS_f(G')$ denote the local sensitivity at distance $t$. 
Also note that $LS_\kcl(G) = \max_{u,v}f_{(k-2)\mathbb{C}}(G(A_{uv}))$.

% using the local variant of sensitivity (LS) $LS_f(G) = max_{G':G\sim G'}|f(G) - f(G')|$. 

\begin{definition}
  \label{def:smooth-upper-bound}
% ($\beta$-smooth bound on LS)
(Smooth sensitivity)
A function $S^{\beta}_f(G)$ is $\beta$-smooth upper bound on the LS if (1) $S^{\beta}_f(G) \geq LS_f(G)$ for all $G\in G^n$ and (2) $S^{\beta}_f(G) \leq e^{\beta}S^{\beta}_f(G')$ for any $G\sim G'\in G^n$.
$S^{*}_{f,\beta}(G) = \max_{t=0,\ldots,{n \choose 2}}e^{-t\beta}LS^{(t)}_f(G)$, the smallest $\beta$-smooth bound function on $LS_f$, is referred to as the smooth sensitivity,
where $\beta$ depends on $\epsilon, \delta$.
\end{definition}

An alternative technique, referred to as ladder functions was developed by~\cite{zhang:sigmod15}.
One can get an $\epsilon$-DP estimate using it, and it is computationally more efficient than smooth sensitivity.

\begin{definition}
  \label{def:ladder-func}
(Ladder function \cite{zhang:sigmod15})
$I_{t}(\cdot)$ is a ladder function of $f$ if (1) $I_{0}(G) \geq LS_f(G)$  for all $G\in G^n$ and (2) $I_{t}(G) \leq I_{t+1}(G')$ for any pair of neighbors $G\sim G'\in G^n$, and any integer $t$.
\end{definition}

\begin{lemma}
\label{lemma:ladder}
(\cite{zhang:sigmod15})
$I_t(G) = \min\left(LS_\kcl(G) + {\ahat + t \choose k -2 } - {\ahat \choose k-2}, GS_\kcl\right)$
is a ladder function for $f_\kcl$, where $\ahat, GS_\kcl$ are defined in Table \ref{tab:notation}.
% where $\ahat = \max_{u,v}a_{uv}$, and $a_{uv}$ is the number of common neighbors of node $u$ and $v$, and $GS_\kcl$ is the global sensitivity of $f_\kcl$.
\end{lemma}

% Defining the local sensitivity at distance $t$ as $LS^{(t)}_f(G) = \max_{G': d(G,G')\leq t}LS_f(G')$, 
% the \emph{smooth sensitivity} $S^{*}_{f,\beta}(G) = \max_{t=0,\ldots,{n \choose 2}}e^{-t\beta}LS^{(t)}_f(G)$ is the smallest $\beta$-smooth bound function on $LS_f$, with $\beta$ being a factor depending on $\epsilon, \delta$.
% Using (exact) smooth sensitivity to calibrate noise, the mechanism $M(G) = f(G) + Lap(2S^{*}_{f,\beta}(G)/\epsilon)$ is $(\epsilon,\delta)$-DP.

\textbf{Approximate notions.}
Standard multiplicative approximation of smooth sensitivity and ladder functions do not ensure DP guarantees~\cite{nguyen2023faster, blocki2022make}.
Nguyen et al.~\cite{nguyen2023faster} showed that an approximation to $S^{*}_{f,\beta}(G)$ can give a private solution. 
In fact, the approximation framework generalizes to all $\beta$-smooth upper bound on local sensitivity. 

% \dungnote{Nguyen et al. approx. framework}

\begin{theorem}
\label{theorem:approx-smooth}
(Privacy via approximate smooth upper bound on local sensitivity)
  Let $\tilde{S}_{f,\beta}$ be a $(\gamma,\delta')$-upper approximation to $S^\beta_f$. For $\gamma=\beta=\frac{16\epsilon}{\log(2/\delta)}$, $M(G) = f(G) + Lap(2\tilde{S}_{f,\beta}/\epsilon)$ is $(\epsilon,\frac{e^{\epsilon/2}+1}{2}\delta + 2 \delta')$-DP.
\end{theorem}

Theorem~\ref{theorem:approx-smooth} applies to functions $f$ with well-defined local sensitivity, i.e., exact, deterministic functions.
For randomized, approximate queries, their local sensitivities may be unbounded, we need to utilize the following mechanism:

\begin{theorem}
  \label{theorem:approx-smooth-approx-query}
Let $A_f$ be an $(\eta,\delta_A)$-approximation to $f$ with $\eta\in(0,1/2)$ , and let $\tilde S$ be a $(\gamma,\delta')$-upper approximation to $S^\beta_{f}$. Suppose $\delta\le 2/e$ and the admissibility budget holds:
\[
  4\eta+\gamma+\beta\;\le\;\frac{\epsilon}{2\ln(2/\delta)}. %\tag{$\star$}
\]
Then
\[
  A(D)=A_f(D)+\mathrm{Lap}\!\left(\frac{1}{\epsilon}\Bigl[\tfrac{4\eta}{1-\eta}A_f(D)+4\tilde S(D)\Bigr]\right)
\]
is $(\epsilon,\ \tfrac{e^{\epsilon/2}+1}{2}\delta+2\delta'+2\delta_A)$-differentially private.
\end{theorem}

However, it is not clear how to compute a good approximation to $S^{*}_{\fkcl,\beta}(G)$.
No approximate versions of ladder functions that preserve privacy are known, and is our focus here.
% We use another technique, referred to as ladder functions \cite{zhang:sigmod15}.

%\dungnote{Nguyen et al. approx. framework}

%A function $\tilde{f}$ is a $(\gamma,\delta)$-upper approximation to $f$ if $f(G) \leq \tilde{f}(G) \leq e^{\gamma}f(G)$ with probability at least $1-\delta$.

%%% Local Variables:
%%% mode: latex
%%% TeX-master: "main"
%%% End:

% \input{approach}
% \input{goal_timeline}
% \input{conclusion}

% \section{$k$-clique Estimation}

\section{Approach}
\label{sec:approach}

% Using the ladder function, one can release the $\eps$-DP variant of $f$ via a noise sampling mechanism~\cite{zhang:sigmod15}.

We start with the ladder function $I_t$ of $\fkcl$, constructed by~\cite{zhang:sigmod15}, defined in Lemma.~\ref{lemma:ladder}.
We then construct an (exact) smooth upper bound on the local sensitivity $S^{\beta}_{\fkcl} = \max_{t=1\ldots T}e^{-t\beta}I_t(G)$, with a carefully chosen threshold $T$, depending on $k$ and $\beta$.
Without it, the range of $t$ must be up to $\binom{n}{2}$, which requires extensive calculation.
We prove that $S^{\beta}_{\fkcl}$ satisfies the smoothness required by the smooth-sensitivity framework, i.e., proving that quantities of two neighbor graphs differ by at most a factor of $e^\beta$ (Lemma~\ref{lemma:smooth-clique}).
% The smoothness property is straightforward for any $t < T$, due to the ladder property ($I_t(G') \leq I_{t+1}(G)$), we can show $e^{-t\beta} I_t(G') \leq e^{-t\beta} I_{t+1}(G) \leq e^{\beta} \cdot e^{-(t+1)\beta} I_{t+1}(G)$, i.e., the $t^{th}$ component of $S^{\beta}_{\fkcl}(G')$ is upper bounded by the $(t+1)^{th}$ component of $S^{\beta}_{\fkcl}(G')$ times $e^\beta$.
% It is more challenging for the last component, i.e., $t=T$, as there is no 

Computing $S^{\beta}_{\fkcl}$ requires two heavily calculated components: $LS_\kcl(G)$ and $\ahat$, the largest count of common neighbors of any two nodes in $G$.
We invoke the approximate smooth sensitivity framework by~\cite{nguyen2023faster}, using an approximation to $S^{\beta}_{\fkcl}$ to calibrate the noise for privacy.
For that purpose, we approximate $LS_\kcl(G)$ and $\ahat$.
Even though an approximation to $\ahat$ has been proposed by~\cite{nguyen2023faster}, we need a stronger guarantee, since in $S^{\beta}_{\fkcl}$, $\ahat$ involves in several binomials in the form of $\binom{\ahat + t}{k-2}-\binom{\ahat}{k-2}$, in which a constant approximation to $\ahat$ is not directly translated to a constant approximation of the binomials.
The other component, $LS_\kcl(G)$, is more challenging to estimate.
It involves counting the number of $k$-cliques that can be changed by adding/deleting an edge in the input graph.
We propose two methods to estimate $LS_\kcl(G)$.
The first method invokes the sublinear approximate count of $k-2$-cliques on each pair-induced subgraph $G(A_{uv})$, then takes the maximum among these counts.
The second method directly estimates the maximum of $k$-cliques involves an edge $(u,v)$ (that can be potentially deleted) and the maximum of $k$-near-cliques missing an edge $(u,v)$ (that can be potentially added), by using a technique called \Turan.
Combining the smooth upper bound construction with the approximation of its components, we design \textsc{FastCliqueDP} that satisfies $(\epsilon,\ddp)$-edge differential privacy with rigorous privacy and runtime guarantees.

% \begin{definition}
%   \label{def:LS-clique}
%   Let $A_{uv}$ be the set of common neighbors of $u$ and $v$.
%   The local sensitivity of $f_\kcl(G)$, $LS_\kcl(G) = \max_{u,v}f_{(k-2)\mathbb{C}}(G(A_{uv}))$, where $G(S)$ is the subgraph induced on $G$ by the node subset $S$, i.e., the count of $(k-2)$-clique in a graph constructed by the common neighbors of $u$ and $v$ together with their edges ($u,v$ are excluded).
% \end{definition}

%\noindent
\textbf{Smooth upper bound on the local sensitivity $S^{\beta}_{\fkcl}$}.

We first start with the ladder function $I_t$ and use it to construct $S^\beta(G)$, which we show is a valid (exact) $\beta$-smooth upper bound on local sensitivity.
While the upper bound on $LS_\kcl(G)$ property is straightforward, due to the appearance of $LS_\kcl$ in $I_t$, the smoothness ($S^\beta(G')\leq e^\beta S^\beta(G)$) requires carefully chosen value of $T$.
In fact, using the property $I_t(G') < I_{t+1}(G)$ for $G\sim G'$, one can prove that the $t^{th}$ component of $S^\beta(G')$ is at most the $(t+1)^{th}$ component of $S^\beta(G)$.
It requires a special analysis at $t=T$, since there is no $(T+1)^{th}$ component in $S^\beta(G)$.
At that point, we have to use $I_T(G)$ as a proxy to $I_{T+1}(G)$, arguing that a hypothetical $I_{T+1}(G)$ can be upper-bounded by $e^{\beta}I_T(G)$.
This involves a careful analysis of the binomials ${\ahat + t \choose k -2 } - {\ahat \choose k-2}$, in particular, the ratio of these at $t=T+1$ over $t=T$ (this is conservative, consider the worst case $LS_\kcl(G)=0$).
We find that when $T$ is sufficiently large, i.e., $T > \frac{(k'-1)e^\beta + 1}{e^\beta - 1}$, the ratio is upper bounded by $e^\beta$ regardless the value of $LS_\kcl(G)$.
When $\beta$ becomes smaller, $T$ grows since the smoothness is more rigid.
$T$ can be upper bounded at $\binom{n}{2}$, as it is the largest distance of two graphs with $n$ nodes.

\begin{restatable}[]{lemma}{smoothclique}
Let $k' = k -2$.
  Choose any constant $\beta$, let $T = \max \left(\left\lceil \frac{(k'-1)e^\beta + 1}{e^\beta - 1} \right \rceil, \binom{n}{2}\right)$ and $S^\beta(G) = \max_{t:0\ldots T}e^{-t\beta}I_t(G)$. $S^\beta(G)$ is a $\beta$-smooth upper bound on $LS_\kcl(G)$. Specifically, for any $G$,  $S^\beta(G) \geq LS_\kcl(G)$, and for any pair of neighbors $G\sim G'$, $S^\beta(G) \leq e^\beta S^\beta(G')$.
 \label{lemma:smooth-clique}
\end{restatable}

\begin{proof}

  Throughout the proof we drop the subscript $\kcl$ since the context of calculating the sensitivity for $\fkcl$ is clear. 

  \textbf{Upper bound of $LS_\kcl(G)$.} By definition, $I_0(G) = LS_\kcl(G)$, and $S^\beta(G) \geq e^{0} I_0(G) = LS_\kcl(G)$.

  \textbf{Smoothness.} We must show $S^\beta(G') \leq e^\beta S^\beta(G)$ for every pair of neighbors $G \sim G'$. Fix any $t \in \{0, 1, \ldots, T\}$, if we can show

  \begin{equation}
    \label{eq:per-t}
    e^{-t\beta} I_t(G') \leq e^{\beta} S^\beta(G),
  \end{equation}

  then $S^\beta(G') = \max_{t=0,1,\ldots T} e^{-t\beta} I_t(G') \leq e^{\beta} S^\beta(G)$.

  Since $I_t$ is a ladder function for $\fkcl$ (Theorem~\ref{theorem:clique-ladder}), $I_t$ satisfies the following property (as Definition~\ref{def:ladder-func}):
  \begin{equation}
    \label{eq:ladder}
    I_t(G') \leq I_{t+1}(G) \quad \text{for all } t \geq 0,
  \end{equation}
  note that the reversed $I_t(G) \leq I_{t+1}(G')$ also holds due to the symmetry of $G$ and $G'$ ($G\sim G' \iff G' \sim G$).

  We consider two cases: (1) $t$ is not at the edge of the range and (2) $t = T$ at the edge.

  \emph{Case 1: $0 \leq t \leq T - 1$.} By Equation~\ref{eq:ladder},
  \begin{align*}
    e^{-t\beta} I_t(G') \leq e^{-t\beta} I_{t+1}(G) 
     = e^{\beta} \cdot e^{-(t+1)\beta} I_{t+1}(G)
    \leq  e^{\beta} S^\beta(G),
  \end{align*}
  where the last inequality holds because $t+1 \leq T$, so $e^{-(t+1)\beta} I_{t+1}(G)$ appears in the max defining $S^\beta(G)$.

  \emph{Case 2: $t = T$.} By Equation~\ref{eq:ladder}, $I_T(G') \leq I_{T+1}(G)$, so it suffices to bound $I_{T+1}(G) / I_T(G) \leq e^\beta$ directly.
  We note that $I_{T+1}(G)$ will not appear in the max defining $S^{\beta}(G)$, and we use it for the purpose of analysis only.
  Writing out the ratio and applying Pascal's rule:
  \begin{align*}
    \frac{I_{T+1}(G)}{I_T(G)}
    &= \frac{LS_\kcl(G) + \binom{\hat a + T + 1}{k'} - \binom{\hat a}{k'}}{LS_\kcl(G) + \binom{\hat a + T}{k'} - \binom{\hat a}{k'}} \\
    &= 1 + \frac{\binom{\hat a + T}{k' - 1}}{LS_\kcl(G) + \binom{\hat a + T}{k'} - \binom{\hat a}{k'}}\\
    &\leq 1 + \frac{\binom{\hat a + T}{k' - 1}}{\binom{\hat a + T}{k'} - \binom{\hat a}{k'}} \text{, as $LS_\kcl(G) \geq 0$,}\\
    &\overset{(a)}\leq  1 + \frac{k'}{T - k' + 1}.
  \end{align*}

  The inequality at $(a)$ is true because using the identity $\binom{\hat a+T}{k'-1} = \frac{k'}{\hat a+T-k'+1}\,\binom{\hat a+T}{k'}$, inequality $(a)$ is equivalent to
  \begin{align*}
    \frac{\binom{\hat a+T}{k'}}{\binom{\hat a+T}{k'} - \binom{\hat a}{k'}}
    \leq
    \frac{\hat a + T - k' + 1}{T - k' + 1}, &\\
    \binom{\hat a+T}{k'}\leq \left({\binom{\hat a+T}{k'} - \binom{\hat a}{k'}} \right)
    \frac{\hat a + T - k' + 1}{T - k' + 1}, &\\
   \binom{\hat a}{k'}\frac{\hat a + T - k' + 1}{T - k' + 1} \leq \binom{\hat a+T}{k'}
    \left(\frac{\hat a + T - k' + 1}{T - k' + 1} - 1\right), &\\
  \binom{\hat a}{k'}\leq \frac{\hat a}{\hat a+T-k'+1}\,
   \binom{\hat a+T}{k'}.&
  \end{align*}
  For $\hat a < k'$ this holds trivially ($\binom{\hat a}{k'} = 0$).
  For $\hat a \geq k'$, it is equivalent to
  $\frac{\binom{\hat a+T}{k'}}{\binom{\hat a}{k'}}
   \geq \frac{\hat a+T-k'+1}{\hat a}$,
  which follows because
  \begin{align*}
    \frac{\binom{\hat a+T}{k'}}{\binom{\hat a}{k'}}
    = \prod_{i=0}^{k'-1}\frac{\hat a+T-i}{\hat a-i}
    \geq \frac{\hat a+T}{\hat a}
    \geq \frac{\hat a+T-k'+1}{\hat a},
  \end{align*}
  where the first inequality uses that every factor with $i \geq 1$ is at least~$1$. 

  \begin{align*}
    \frac{I_{T+1}(G)}{I_T(G)}
    &\leq  1 + \frac{k'}{T - k' + 1}\\
    &\leq  1 + \frac{k'}{\frac{(k'-1)e^\beta + 1}{e^\beta - 1} - k' + 1}, \text{(Lemma's condition for $T$)}\\
    &\leq 1 + \frac{k'(e^\beta-1)}{(k'-1)e^\beta + 1 - (k'-1)(e^\beta-1)}\\
    &= 1 + \frac{k'(e^\beta-1)}{k'}=e^\beta.
  \end{align*} 

  Finally, we have $e^{-T\beta}I_T(G') \leq e^{-T\beta}I_{T+1}(G) \leq e^{-T\beta}e^{\beta}I_T(G) \leq e^{\beta}S^{\beta}(G)$, where the first inequality is due to the ladder property of $I_t(G)$, the second inequality as above, and the last inequality is due to $e^{-T\beta}I_T(G) \leq S^{\beta}(G)$ by its definition, and the Lemma follows.

\end{proof}

Applying framework \textbf{Approximate smooth sensitivity for exact queries} above, we propose an efficient design of $\tilde{S}^\beta_\kcl(G)$ (we drop the subscript $\kcl$ when the context is clear). 

\begin{restatable}[]{theorem}{approxsmooth}
  With $k' = k-2$, let $\tilde{a}$ be a $(\frac{\gamma}{2k'}, \delta_{\hat{a}})$-upper approximation of $\hat{a}$, $\LStil(G)$ be a $(\gamma, \delta_{\kcl})$-upper approximation of $LS_\kcl(G)$ . 
  Let $\Itil_t(G)$ be defined as:
  \begin{align}
    \widetilde{I_t}(G) = \min\left(\LStil(G) + {{\tilde{a}+ t}\choose{k-2}} - {\tilde{a} \choose {k-2}} , GS\right)
  \end{align}
The $(\gamma,\delta_{\ahat}+\delta_\kcl)$-approximate smooth upper-bound $\tilde{S}^\beta(G)$ of the local sensitivity is defined as:
$\tilde{S}^\beta(G) = \max_{t:0\ldots T}e^{-t\beta}\tilde{I}_t(G)$
  %\label{def:smooth-bound-approx}
\label{theorem:tildeS}
\end{restatable}

\begin{proof}
Lemma~\ref{lemma:I-approx} shows that for each $t: \widetilde{I_t}(G)$ is a $(\gamma, \delta_\kcl + \delta_{\hat{a}})$-upper approximation of $I_t(G)$. Taking the union bound on the failure probabilities of approximating $LS_\kcl(G)$ and $\ahat$, $\delta_\kcl$ and $\delta_\ahat$ respectively as they are computed once, and are reused across values of $t$, the Theorem follows.
\end{proof}

\begin{lemma}
$\widetilde{I_t}(G)$ is a $(\gamma, \delta_\kcl + \delta_{\hat{a}})$-upper approximation of $I_t(G)$.
\label{lemma:I-approx}
\end{lemma}

\begin{proof}
  Let $k' = k - 2$.
  Let $B(x,t) = \binom{x + t}{k'} - \binom{x}{k'}$.

  \textbf{First}, we show that $I_t(G) \leq \Itil(G)$, or $LS(G) + B(\ahat, t) \leq \LStil(G) + B(\atil, t) $.
  As we have $LS(G) \leq \LStil(G)$ (stated by Theorem~\ref{theorem:tildeS}), we need to show $B(\ahat, t)\leq B(\atil, t)$, knowing that $\ahat \leq \atil$ (stated by Theorem~\ref{theorem:tildeS}).
  Applying the telescoping identity, we have:
  \begin{align*}
    B(x,t) = \binom{x + t}{k'} - \binom{x}{k'} = \sum_{i=0}^{t-1}\binom{x + i}{k' -1},
  \end{align*}
  where each $\binom{x+i}{k'-1}$ is non-decreasing in $x$, so $B(x,t)$ is non-decreasing in $x$.
  It yields $B(\ahat, t) \leq B(\atil, t)$ since $\ahat \leq \atil$.

  \textbf{Second}, we show that $\Itil_t(G) \leq e^\gamma I_t(G)$.
  Using Lemma~\ref{lemma:ahat-atil}, we have the bound $\LStil(G) + B(\atil, t)\leq e^\gamma LS(G) + e^\gamma B(\ahat, t)$.
  We have:
  \begin{align*}
    \Itil_t(G) &= \min(\LStil(G) + B(\atil, t), GS_\kcl)\\
               &\leq \min(e^\gamma(LS(G) + B(\ahat, t)), GS_\kcl) \\
               &\leq \min(e^\gamma(LS(G) + B(\ahat, t)), e^\gamma GS_\kcl)
    = e^\gamma I_t(G).
  \end{align*}
  Taking the union bound on the failure probabilities of $\atil$ and $\LStil(G)$, the arguments above hold simultaneously with probability at least $1-\delta_\ahat - \delta_\kcl$, and the Lemma follows.
\end{proof}

In~\textsc{FastCliqueDP} (Algorithm~\ref{alg:clique-dp-fast}), for simplicity, we use $\ddp$ (the $\delta$-part of privacy) for both the failure probabilities of the approximation (Step 1) and the privacy (Step 2).
This makes the $\delta$-part of the final privacy guarantee slightly larger ($\frac{e^{\eps/2}+5}{2}\ddp$), with negligible effects for typical values of $\ddp = O(n^{-1})$.
Practical algorithms may want to use separate parameters with proper scaling to achieve desired values of $\ddp$.

%\begin{wrapfigure}{L}{0.6\textwidth}
%  \begin{minipage}{0.6\textwidth}
    \begin{algorithm}[H]
      \caption{\textsc{FastCliqueDP} \\ 
        \textbf{Input:} $G, k, \eps, \ddp$ \\
        \textbf{Output:} An $(\eps,\frac{e^{\eps/2}+5}{2}\ddp)$-DP estimation of $k$-clique in $G$}
      \begin{algorithmic}[1]
        \STATE $\delta_\ahat = \delta_\kcl = \ddp/2$ 
        % \STATE $x = \frac{\log{(1+k)}\log{2/\delta}}{2\eps}$
        \STATE $\gamma = \beta = \frac{16\eps}{\log{(2/\ddp)}}$ 
        \STATE $k' = k-2$
        \STATE $\atil = \textsc{FastA}(G, \gamma/(2k'), \delta_\ahat)$
        \STATE Estimate $LS_\kcl(G)$ for $(\gamma,\delta_\kcl)$-upper approximation using Algorithm \textsc{Clique.Approx} (Section~\ref{sec:blackbox}) or \textsc{Max.Clique.Pair} (Section \ref{sec:direct})
        %\STATE Estimate $LS_\kcl(G)$ for $(\gamma,\delta_\kcl)$-upper approximation // Either by Section~\ref{sec:blackbox} or~\ref{sec:direct}
        \STATE Calculate $\tilde{I}_t(G)$ as defined in Thm.~\ref{theorem:tildeS} using $LS_\kcl(G)$ 
        \STATE $T =\left\lceil \frac{(k'-1)\,e^\beta + 1}{e^\beta - 1} \right\rceil$
        \STATE $\tilde{S}^\beta(G) = \max_{t:0\ldots T}e^{-t\beta}\tilde{I}_t(G)$
        \STATE \textbf{Return} $f_\kcl(G) + Lap(\frac{2\tilde{S}^\beta(G)}{\eps})$
      \end{algorithmic}
      \label{alg:clique-dp-fast}
    \end{algorithm} 
%   \end{minipage}
% \end{wrapfigure}

The key analysis of the proof of Theorem~\ref{theorem:tildeS} is to prove $\Itil_t(G)$ is a $(\gamma,\ddp)$-upper approximation of $I_t(G)$ for all values of $t$.
While $t$ and its upper bound $T$ have no effect on $\LStil(G)$, they dictate the approximation factor of $\atil$.
Informally, since ${\tilde{a} \choose {k-2}} = {\tilde{a} \choose {k'}} \sim \frac{\tilde{a}^{k'}}{k'!}$, we can choose the approximation factor for $\ahat$ to be $\gamma/{2k'}$, such that $\hat{a} \leq \tilde{a} \leq e^{\gamma/{(2k')}}\hat{a}$ .
It yields ${\hat{a} \choose {k-2}} \leq{\tilde{a} \choose {k-2}} \leq e^\gamma{\hat{a} \choose {k-2}}$.
In Lemma~\ref{lemma:I-approx}, we prove a stronger, rigorous guarantee for $I_t$.
Combining with the result of Lemma~\ref{lemma:local-approx} (approximation guarantee for $\LStil(G)$),  $\widetilde{I_t}(G)$ is a $(\gamma, \delta_\kcl + \delta_{\hat{a}})$-upper approximation of $I_t(G)$.

There is a subtle detail for the approximation of $\ahat$.
Lemma~\ref{lemma:ahat-atil} requires a mild assumption that $\ahat \geq 2(k'-2)$, which almost all realistic graphs satisfy for practical $k$-cliques.
To handle this case rigorously, the algorithm can use $\atil$ as a proxy to $\ahat$ and compare to $2(k'-2)$ as a fallback mechanism.
If $\atil/e^{\gamma/(2k')} \geq 2(k'-2)$, then $\ahat \geq 2(k'-2)$ and the algorithm proceeds as described.
If $\atil/e^{\gamma/(2k')} < 2(k'-2)$, we compute $\ahat$ exactly for $\Itil_t(G)$, which takes $O(kn^2)$, and lose the efficiency of the quick estimation $\atil$.
In this case, we show that $\Itil_t(G)$ is a valid upper-approximation of $I_t(G)$ (Lemma \ref{lemma:I-approx}), as $\LStil(G)$ is the only approximation components of $\Itil(G)$ and the rest is calculated exactly, and the privacy is guaranteed.

With $\LStil(G)$ and $\atil$, we are ready to construct $\Itil_t(G)$, and subsequently, $\tilde{S}^{\beta}_\kcl(G)$.
We only need to compute $\LStil(G)$ and $\atil$ once, since the differences between $I_t$s are the binomials constructed by $\atil$ and $t$.
We utilize the approximation smooth sensitivity, using $\tilde{S}^{\beta}_\kcl(G)$ to calibrate a Laplacian noise for $(\epsilon,\ddp)$-DP.
The failure probability of $\tilde{S}^{\beta}_\kcl(G)$ to approximate $S^\beta_\kcl(G)$ will be incorporated into the $\delta$-part of the final DP guarantee.

%\begin{theorem}
%\label{theorem:tildeS}
%$\tilde{S}_{\beta}(G) = \max_{t:0\ldots T}e^{-t\beta}\tilde{I}_t(G)$ is a $(\gamma,\ddp)$-upper approximation to $S_{\beta}(G)$, with $\ddp = \delta_\kcl + \delta_{\hat{a}}$.
%\end{theorem}

\begin{restatable}[]{theorem}{privacy}
  \label{theorem:privacy}
  Algorithm~\ref{alg:clique-dp-fast} is $(\eps,\frac{e^{\eps/2}+5}{2}\ddp)$-differentially private.
\end{restatable}

\begin{proof}
 With $\widetilde{I_t}(G)$ as being defined in Theorem~\ref{theorem:tildeS}, it is an $(\gamma,\delta)$-upper approximation of $I_t(G)$, with $\gamma = \frac{16\eps}{\log{2/\ddp}}$, and the failure probability $\ddp = \delta_{k\mathbb{C}} + \delta_{\ahat}$, due to the union bound of the failure probabilities of estimating $\ahat$ and $LS_\kcl(G)$.
 Let $S^{\beta}(G)$ being as defined in Lemma~\ref{lemma:smooth-clique}, it is clear that $\tilde{S}^{\beta}(G)$ is a $(\gamma,\ddp)$-upper approximation of $S^{\beta}(G)$.
 Applying the mechanism of Theorem~\ref{theorem:approx-smooth} with the approximation of $S^{\beta}(G)$ and $\delta' = \ddp$(we set $\delta'$ in the context of Theorem~\ref{theorem:approx-smooth} to $\ddp$ for brevity, practical algorithms may choose a different value), we have $f_\kcl(G) + Lap(\frac{2\tilde{S}^\beta(G)}{\eps})$ satisfies $(\eps, \frac{e^{\eps/2}+5}{2}\ddp )$-differential privacy.
\end{proof}

\textbf{For approximate count of $k$-clique.}

Let $(\eta,\delta_A)$-Approximation $A_\fkcl$ to $\fkcl$, i.e., with probability at least $1-\delta_A$, $A_\fkcl(G) \in (1\pm \eta)\fkcl(G)$, with the probability space over the randomness of $A_\fkcl$.

\begin{algorithm}[H]
  \caption{\textsc{FastApproxCliqueDP} \\ 
    \textbf{Input:} $G, k, \eps, \ddp$, $(\eta,\delta_A)$-Approximation $A_\fkcl$ to $\fkcl$ \\
    \textbf{Output:} An $(\eps,\frac{e^{\eps/2}+5}{2}\ddp)$-DP estimation of $k$-clique in $G$}
  \begin{algorithmic}[1]
    \STATE $\delta_\ahat = \delta_\kcl = \ddp/2$ 
    % \STATE $x = \frac{\log{(1+k)}\log{2/\delta}}{2\eps}$
    \STATE $\gamma = \beta = \frac{\eps}{8\log{(2/\ddp)}}$ 
    \STATE $k' = k-2$
    \STATE $\atil = \textsc{FastA}(G, \gamma/(2k'), \delta_\ahat)$
    \STATE Estimate $LS_\kcl(G)$ for $(\gamma,\delta_\kcl)$-upper approximation using Algorithm \textsc{Clique.Approx} (Section~\ref{sec:blackbox}) or \textsc{Max.Clique.Pair} (Section \ref{sec:direct})
    %\STATE Estimate $LS_\kcl(G)$ for $(\gamma,\delta_\kcl)$-upper approximation // Either by Section~\ref{sec:blackbox} or~\ref{sec:direct}
    \STATE Calculate $\tilde{I}_t(G)$ as defined in Thm.~\ref{theorem:tildeS} using $LS_\kcl(G)$ 
    \STATE $T =\left\lceil \frac{(k'-1)\,e^\beta + 1}{e^\beta - 1} \right\rceil$
    \STATE $\tilde{S}^\beta(G) = \max_{t:0\ldots T}e^{-t\beta}\tilde{I}_t(G)$
    \STATE $\eta = \frac{\eps}{16\log{(2/\ddp)}}$
    \STATE \textbf{Return} $A_\fkcl(G) + Lap(\frac{\frac{4\eta}{1-\eta}\eta A_\fkcl(G) + 4\tilde{S}^\beta(G)}{\eps})$
  \end{algorithmic}
  \label{alg:clique-approx-dp-fast-approx}
\end{algorithm} 

\begin{restatable}[]{theorem}{privacy-approx}
  \label{theorem:privacy-approx}
  Algorithm~\ref{alg:clique-approx-dp-fast-approx} is $(\eps,\frac{e^{\eps/2}+5}{2}\ddp + 2\delta_A)$-differentially private.
\end{restatable}

The proof of Theorem~\ref{theorem:privacy-approx} is similar to proof of Theorem~\ref{theorem:privacy}, except that we invoke Theorem~\ref{theorem:approx-smooth-approx-query} for the approximate count $A_\fkcl$ of $k$-clique.

\textbf{Remark.}
We define and analyze $S^{\beta}_{\fkcl}$ as a valid smooth upper bound on the local sensitivity of $\fkcl$.
The exact calculation of $S^{\beta}_{\fkcl}$ involves the local sensitivity $LS_\kcl(G) =\max_{u,v}f_{(k-2)\mathbb{C}}(G(A_{uv}))$, which in turn calculates the counts of clique of smaller size ($k-2$).
Algorithm~\ref{alg:clique-dp-fast} specifies the required steps to calculate $\tilde{S}^{\beta}_{\kcl}$ approximate $S^{\beta}_{\kcl}$.
The major component we need to approximate to $LS_\kcl(G)$ (See Section~\ref{sec:fastA} for the approximation of $\ahat$).
There are two possible approaches.
In the first approach (Section~\ref{sec:blackbox}), we apply a fast, accurate clique counting algorithm that can give us the utility and failure bound in a black-box manner for each local neighborhood of a pair of vertices $u, v$.
In the second approach (Section~\ref{sec:direct}), we design a novel dedicated algorithm to estimate $\max_{u,v}f_{(k-2)\mathbb{C}}(G(A_{uv}))$ directly.
The different approaches give different runtime complexities, but they achieve the same required approximation guarantee for the estimation of $LS_\kcl(G)$, hence the privacy of Theorem~\ref{theorem:privacy} applies to both.

\noindent
\subsection{Algorithm \textsc{Clique.Approx} for black-box approximation of  $LS_\kcl(G)$.}
\label{sec:blackbox}

By definition, $LS_\kcl(G) = \max_{G': G'\sim G}|\fkcl(G') - \fkcl(G)|$, it indicates the largest difference in $\#k$-cliques of a neighbor $G'$ from $G$, in which $G$ and $G'$ differ in exactly one edge.
Assume that edge $(u,v)$ exists, a specific $k$-clique that involves $(u,v)$ contains $k-2$ nodes that are common neighbors of $u$ and $v$, and that they are all connected to each other, hence they form a $k-2$-clique if we remove $u$ and $v$.
Therefore, one way to calculate the number of $k$-clique involving $u,v$ is to count all $k-2$-cliques in a subgraph induced by the common neighbors of $u$ and $v$.
We can utilize any clique counting (with clique size $k-2$) algorithm for each induced subgraph, hence the approach is named ``black-box''.
Our algorithm \textsc{Clique.Approx} uses \algK~\cite{eden2018approximating}, that estimates $k$-clique counts in sublinear time (in specific regimes).
However, \algK~has failure probability of $1/3$, where our $\LStil(G)$ requires a failure probability much lower $O(1/n)$ as it will be incorporated into the $\delta$-part of the privacy.
We then have to sample $O(\log{n})$ instances of $\algK$, and takes the median in order to achieve the desired failure probability.
The main results are summarized below.
The full details, including the description of \textsc{Clique.Approx} are presented in Section~\ref{sec:blackbox-appendix}.

\begin{restatable}[]{lemma}{localapprox}
$\widetilde{LS}(G) = \max_{u,v}\textsc{Clique.Approx}(G(A_{uv}), k-2, \gamma, \delta_{\kcl})$ is a $(\gamma,\delta_{\kcl})$-upper approximation of $LS_\kcl(G)$.
  \label{lemma:local-approx}
\end{restatable}

\begin{restatable}[]{theorem}{totaltimeblackbox}
\label{theorem:total-time-blackbox}
The calculation of $\tilde{S}_{\beta}(G)$ using \textsc{Clique.Approx} takes
\begin{align*}
  O\Bigg( \Big[ \sum_{u,v \in V(G)}\mathbf{T}_{\operatorname{\textsc{Clique.Approx}}}(G(A_{uv}), k-2) + m \\+ n + \frac{||W||_1}{\hat{a}^2}\Big] poly(\log{m}, \log{n}, 1/\alpha, k, \log{1/\delta})\Bigg)
\end{align*}  expected running time, with $W$ as defined in Definition~\ref{def:W}.
\end{restatable}

\noindent
\subsection{Algorithm \textsc{Max.Clique.Pair} for direct estimation of $LS_\kcl(G)$.}
\label{sec:direct}

$LS_\kcl(G)$ is the maximum count of $k$-cliques that changes when we flip the connection between a pair of vertices $u,v\in V(G)$.
If we flip $(u,v)\in E(G)$, i.e., delete the edge $(u,v)$, we remove all $k$-cliques that contain $(u,v)$.
If we flip $(u,v)\notin E(G)$, i.e., add a new edge $(u,v)$, we transform all near-cliques missing $(u,v)$ to be new $k$-cliques in $G'$.
For $(u,v)\in E(G)$, let $K_{u,v}$ denote the count of $k$-cliques containing $(u,v)$, and for $(u,v)\notin E(G)$, let $\Kone_{u,v}$ denote the count of $k$-near-cliques missing $(u,v)$.
Let $\Khat = \max_{(u,v)\in E(G)}K_{u,v}$ and $\Konehat = \max_{(u,v)\notin E(G)}\Kone_{u,v}$.
Therefore $LS_\kcl(G) = \max_{u,v} |\text{change of $\#k$-cliques flipping (u,v)}| = \max(\Khat, \Konehat)$.
We first describe the technique to estimate $\Khat$ (for the cliques).
The estimation of near-cliques are built upon clique estimation.
We utilize a technique called \Turan~from~\cite{10.1145/3038912.3052636, jain2020provably}, which is based on an observation that when the edge density of a subgraph is high enough, it is likely to contain a clique.
In this, we construct a set of shadows, each contains a dense subgraph.
The shadow has a special property, that there exists a one-to-one mapping between a (potentially) small clique in the shadow to one $k$-clique in the original graph. 
This is used to create a shadow sampling method that all $k$-cliques (via mapping) can be sampled with uniform probability $1/w(\Ess_k)$, in which $w$ is defined in Algorithm~\ref{alg:max-pair-clique-approx}.
We present the fundamentals of \Turan~in Section~\ref{sec:turan-shadow}.

The original setup of \Turan~ estimates the total count of cliques in a graph, and we have to extend it to estimate the clique counts per each individual edge.
Even with the estimation of $K_{u,v}$ per each edge, it is not straightforward to find $\Khat = \max_{u,v} K_{u,v}$, because the counts of different pairs $(u,v)$ may have vastly different variances, and they do not converge to their expectations the same way, e.g., lower $K_{u,v}$ counts require more samples to be accurate.
We design a generic routine in Section~\ref{sec:point-statistics}, to estimate the maximum of point-wise statistics $\max_{j\in J} X_j$ accurately and with high probability, in which $\Xtil$ is an estimator that returns $\Xtil_j$ for each $j\in J$.
We note that $X_j$s may require different samples of $\Xtil$ to converge within their expected values.
The routine applies to any estimator $X$ that is point-wise accurate, defined in Definition~\ref{def:estimator-accurate}.

\subsubsection{Maximum of  Point-wise statistic estimation}
\label{sec:max-point-wise}

\begin{restatable}{definition}{pointwiseaccurate}
  \label{def:estimator-accurate}
  A (stochastic) estimator $\Xtil$ is $(\theta,\delta)$-point-wise accurate if for every $i\in J$, once the number of sampling iteration $s \geq \frac{W\log{2/\delta}}{\theta^2X_{i}}$ for some scaling factor $W$ that is independent on $i$, with probability at least $1-\delta$, the estimator $\Xtil_i$ is within $(1\pm\theta)$-multiplicative factor of $X_i$, i.e., $\Pr[\Xtil_i\notin (1\pm\theta)X_i] \leq \delta$.
\end{restatable}

We observe that finding $\max_{j\in J} X_j$ is easier if we know a lower bound $\tau < \max_{j\in J} X_j$.
With $\tau$, we can set the number of sampling iteration $s \geq \frac{W\log{2/\delta}}{\theta^2\tau}$, and hence guarantees that, $\Xtil_{j^*}: X_{j^*} = \max_{j\in J} X_j$ is accurate.
Beside that, for many $j$ such that $X_j << X_{j^*}$, their estimator $\Xtil_j$ might not be accurate, but still end up far smaller than  $\Xtil_{j^*}$.
However, we do not know $\tau$, and setting it too low yields too many iterations and reduces efficiency.
We design a routine to guess $\tau$ in Algorithm~\ref{alg:tau-estimation}, with the initial guessed value to be the largest possible value of $\tau$.
We use $\tau$ to calculate the number of iterations, and reduce its value if the observations of sampling values do not match our expectation (that $\tau \lessapprox \max_{j\in J} X_j$.
Algorithm~\ref{alg:tau-estimation} guarantees that we will find such $\tau$, and we use its value to accurately estimate the maximum of point-wise statistics (Theorem~\ref{lemma:stat-estimation}).

\subsubsection{Maximum of pair-wise clique count}
\label{sec:pair-clique}

Next, we have to prove that our estimation of $K_{u,v}$ satisfies the requirements in Definition~\ref{def:estimator-accurate}.
We observe that for each pair $u,v$, a clique containing $(u,v)$ is sampled with probability proportional to $K_{u,v}$, due to the uniformity of clique sampling. 
However, the events that cliques involving different pairs $(u,v)$ are sampled in each iteration are not independent across different pairs of vertices, and therefore a straightforward Chernoff bound cannot apply here.
We need to use a proxy count per edge $(u,v)$, per iteration, and show that across iterations, these counts are independent.
The concentration bound therefore states that if we run enough iterations, the count of cliques containing $(u,v)$ is close to its expectation, satisfying the accuracy requirements.

% for clique 
For a pair of vertices $u$ and $v$, a clique $K$ is a pair-wise clique of the pair $u, v$ if $(u, v) \in E(K)$, i.e., $u$ and $v$ are two vertices in the clique $K$.
We use $C^k_{{u,v}}$ to denote the set of all pair-wise cliques (of a fixed size $k$) of the pair $u, v$, and let $K^k_{u,v} = |C^k_{{u,v}}|$.
For simplicity, we will drop the clique size $k$ when it is obvious from the context.
We assume that each unique clique has a unique identifying index $j$, denoted as $K_j$.
The unique clique identifier $K_j$ and its variables $X_{j,t}$ are for the purposes of analyzing the utility, the actual algorithm does not need to implement these.
We design an algorithm to estimate $\Khat = \max_{(u,v)\in E(G)} K_{u,v}$, i.e., the maximum count of pair-wise cliques among all pairs of nodes in a graph $G$.

\begin{algorithm}[ht]
  \caption{\textsc{Unbounded.PairWise.Clique.Approx} \\ 
    \textbf{Input:} $G, s$ and \Turan~$\Ess_k$ \\
    \textbf{Output:} An estimation of $ K_{u,v}: \forall (u,v)\in E(G)$}

  \begin{algorithmic}[1]
    \STATE $w(\Ess_k) := \sum_{(P, S, l)\in \Ess_k}{|S| \choose l}$
    \STATE $\forall (u,v)\in E(G): Z_{u,v} := 0$
    \FOR{$t := 1,\ldots, s$}
    \STATE $\forall j: X_{j,t} := 0$
    \STATE $\forall (u,v)\in E(G): Z_{u,v,t} := 0$
    \STATE $H := \Sample(\Ess_k)$ (Algorithm~\ref{alg:shadow-sample})
    \IF{$\exists j: H = K_j$ (H is a clique)}
    \STATE $X_{j,t} := 1; \forall j'\neq j: X_{j',t} := 0$ %\text{Note: For the purpose of analysis only.}
    \FOR{$(u,v) \in E(H)$ (each edge in $H$)}
    \STATE $Z_{u,v,t} := 1$
    \ENDFOR
    \ENDIF
    \STATE $\forall (u,v) \in E(G): Z_{u,v} := Z_{u,v}+ Z_{u,v,t}$
    \ENDFOR
    \STATE $\forall (u,v) \in E(G): \widetilde{K}_{u,v} := Z_{u,v}\frac{w(\Ess_k)}{s}$ 
    \STATE \textbf{Return }  $\widetilde{K}$
  \end{algorithmic}

  \label{alg:pair-clique-approx}
\end{algorithm} 

\begin{lemma}
  \label{lemma:Kuv-expectation}
  For every $(u,v)\in E(G)$, $\E[Z_{u,v}\frac{w(\Ess_k)}{s}] = K_{u,v}$.
\end{lemma}

\begin{proof}
  By its definition, $Z_{u,v} = \sum_{t:1\ldots s}Z_{u,v,t} = \sum_{t:1\ldots s}\sum_{j:(u,v)\in E(K_j)}X_{j,t}$,  where the last equality is because $Z_{u,v,t}=1 \iff \sum_{j:(u,v)\in E(K_j)}X_{j,t} = 1$, i.e., one of the cliques containing $(u,v)$ is sampled.

  Theorem~\ref{theorem:shadow-sample} shows that any clique $K_j$ can be sampled from $\Ess_k$~with probability $\frac{1}{w(\Ess_k)}$ by Sample$(\Ess_k)$, hence $X_{j,t}$ is a Bernoulli variable with $\Pr[X_{j,t} = 1] = \frac{1}{w(\Ess_k)}$ and $\E[X_{j,t}] = \frac{1}{w(\Ess_k)}$.
  By the linearity of expectations, $\E[Z_{u,v}] = \sum_{t:1\ldots s}\sum_{j:(u,v)\in E(K_j)}\E[X_{j,t}] = \sum_{t:1\ldots s}\sum_{j:(u,v)\in E(K_j)}\frac{1}{w(\Ess_k)} = \frac{s}{w(\Ess_k)}K_{u,v}$, where the last equality is because $K_{u,v} = |\{j: (u,v)\in E(K_j)\}|$, and the Lemma follows.
\end{proof}

\begin{lemma}
  \label{lemma:Kuv-whp-bound}
  Given a constant $\theta$ (the approximation factor), and a failure probability $\delta$, if the number of samples $s \geq \frac{3w(\Ess_k)\log{2/\delta}}{\theta^2K_{u,v}}$, then with probability at least $1-\delta$, $Z_{u,v}\frac{w(\Ess_k)}{s} \in (1\pm\theta)K_{u,v}$.
\end{lemma}

\begin{proof}
  For a fixed $t$, let $Y_{u,v,t} = \sum_{j:(u,v)\in E(K_j)}X_{j,t}$.
  It is clear to see that $Y_{u,v,t} = 1$ if and only if a clique containing $(u,v)$ is sampled at step $t$.
  Across different $t$, $Y_{u,v,t}$s are independent, with $\E[Y_{u,v,t}] = \E[\sum_{j:(u,v)\in E(K_j)}X_{j,t}] = \frac{K_{u,v}}{w(\Ess_k)}$.
  Since $Z_{u,v} = \sum_{t:1\ldots s}\sum_{j:(u,v)\in E(K_j)}X_{j,t} =\sum_{t:1\ldots s}Y_{u,v,t}$, with $Y_{u,v,t}$ being a Bernoulli variable with $E[Y_{u,v,t}] = \frac{K_{u,v}}{w(\Ess_k)}$, using the Chernoff bound, we have:
 \begin{align*}
   \Pr[Z_{u,v} \notin (1\pm\theta) \E[Z_{u,v}]] &\leq 2\exp(-\frac{\theta^2\E[Z_{u,v}]}{3})\text{, }\\
   \Pr[Z_{u,v} \notin (1\pm\theta) K_{u,v}\frac{s}{w(\Ess_k)}] &\overset{(a)}{\leq} 2\exp(-\frac{\theta^2K_{u,v}{s}}{3{w(\Ess_k)}}),\\
   \Pr[Z_{u,v}\frac{w(\Ess_k)}{s} \notin (1\pm\theta) K_{u,v}] &\overset{(b)}{\leq} 2\exp(-\log{2/\delta}), \\
    &\leq \delta, \\
 \end{align*} 
 where $(a)$ is because of Lemma~\ref{lemma:Kuv-expectation}, and $(b)$ is because of substituting $s$,
 and the Lemma follows.
\end{proof}

\begin{corollary}
 \label{cor:clique-point-wise-accurate}
 Let $J = \{(u,v):(u,v)\in E(G)\}$, given a constant factor $\theta$, failure probability $\delta$, \CliqueAlg~is $(\theta,\delta)$-point-wise accurate.
\end{corollary}

After proving that $\CliqueAlg$ satisfies the properties of a point-wise statistic, we apply the routines stated in the Section~\ref{sec:max-point-wise}.

%NOTE: Explain choosing W = 3w{S}

\begin{algorithm}[ht]
  \caption{\textsc{Max.PairWise.Clique.Approx} \\ 
    \textbf{Input:} $G, k, \theta, \delta$ \\
    \textbf{Output:} An  approximation of $\max_{(u,v\in E(G))} K_{u,v}$}

  \begin{algorithmic}[1]
    \STATE $\Ess_k := \Turan(G, k)$
    \STATE $w(\Ess_k) := \sum_{(P,S,l)\in \Ess_k}{|S|\choose l}$
    \STATE $\tau_t := Algorithm~\ref{alg:tau-estimation}\big($ 

     \hspace{\algorithmicindent}
     \CliqueAlg, 

     \hspace{\algorithmicindent}
     ${\max_{v\in V(G)}\degr{v} \choose k-2},3w(\Ess_k),\delta\big)$

     \STATE \textbf{Return} $Algorithm~\ref{alg:stat-estimation}\big($

     \hspace{\algorithmicindent}
     \CliqueAlg,

     \hspace{\algorithmicindent}
     $3w(\Ess_k), \theta, \delta, \tau_t\big)$
  \end{algorithmic}

  \label{alg:max-pair-clique-approx}
\end{algorithm}

\begin{theorem}
 \label{theorem:max-pair-clique-utility}
 With probability at least $1-4\delta$, \textsc{Max.PairWise.Clique.Approx} returns $\max_{(u,v)\in E(G)}\Ktil_{u,v} \in (1\pm\theta)\Khat$.
\end{theorem}

\begin{proof}
  Corollary~\ref{cor:clique-point-wise-accurate} states that \textsc{Max.PairWise.Clique.Approx} is $(\theta,\delta)$-point-wise accurate regarding the estimation of $K_{(u,v)}$, therefore we can apply the Point-wise statistic estimation in Section~\ref{sec:point-statistics}. Applying Lemma~\ref{lemma:binary-search-tau}, we have $\tau_t$ in Algorithm~\ref{alg:max-pair-clique-approx} satisfies that $\tau_t \leq \Khat \leq 4\tau_t$ with probability at least $1-\delta$. Applying Lemma~\ref{lemma:stat-estimation}, with $\tau_t$ satisfying the requirements of the Algorithm~\ref{alg:stat-estimation}, the output of $Algorithm~\ref{alg:stat-estimation}$ with the first parameter of the input being $\CliqueAlg$ and the rest of the input being $3w(\Ess_k), \theta, \delta, \tau_t$ is $\in (1\pm\theta)\Khat$ with probability at least $1-3\delta$. Taking the union bound on the failure probabilities of Algorithm~\ref{alg:tau-estimation} and Algorithm~\ref{alg:stat-estimation}, the Theorem follows.
\end{proof}

\subsubsection{Maximum of pair-wise near-clique count.}
\label{sec:pair-near-clique}

A near-clique missing $(u,v), (u,v) \notin E(G)$ is a subgraph $H$ induced from $G$, such that $|V(H)| = k$, $u, v\in V(H)$, and the hypothetical graph $H': V(H') = V(H), E(H') = E(H) + (u,v)$ is a $k$-clique.
For near-clique missing an edge $(u,v)$, we need to modify the method.
Each $k$-near-clique missing $(u,v)$ contains exactly two complement $k-1$-cliques sharing $k-2$ common vertices, one contains $u$ and does not contain $v$, while the other contains $v$ and not $u$.
Therefore, the main idea is that if we find one of the $k-1$-clique, we can find the other and subsequently, construct the near-clique containing them.
Suppose we detect a $k-1$-clique $K_j$, and are finding its complement $K_i$s (they may be multiple of them).
We know that $\exists v: \{v\} = K_i \setminus K_j$.
We therefore search for $v$ among all vertices adjacent to $K_i$, checking if they are connected to $k-2$ vertices in $K_j$.
To improve efficiency, we can limit the search space to the union of neighbors of two lowest degree nodes in $K_j$, since such $v$ must be connected to at least one of them.
For a specific $k$-near-clique, the expectation it is sampled each iteration is the sum of expectations of its complement $k-1$-cliques, which is twice the expectation of an arbitrary $k-1$-clique.

\medskip
Let $\bar{C}^{k}_{u,v}$ be the set of near-cliques missing $(u,v)$.
Let $\Kone_{u,v}$ denote the count of near-cliques missing $(u,v)$, i.e., $\Kone_{u,v: (u,v)\notin E(G)} = |\bar{C}^{k}_{u,v}|$, and $\Konehat = \max_{u,v: (u,v)\notin E(G)}\Kone_{u,v}$ be the maximum count of near-cliques missing the same edge.
We present below the algorithm to estimate $\Konehat$.
We calculate the estimation for all $u < v, (u,v) \notin E(G)$, as the same statistic for $v < u$ can be easily inferred.
Let $N^*(H)$ denote the set containing all nodes adjacent to the subgraph $H$.
The unique near-clique identifier $\Kone_j$ and its variables $X_{j,t}$ are for the purposes of analyzing the utility, the actual algorithm does not need to implement these.

\begin{algorithm}[ht]
  \caption{\textsc{Unbounded.PairWise.NearClique.Approx} \\ 
    \textbf{Input:} $G, s$ and Turan shadow $\Ess_{k-1}$ with clique size $k-1$ \\
    \textbf{Output:} An  approximation of $ \Kone_{u,v}: \forall (u,v)\notin E(G)$}

  \begin{algorithmic}[1]
    \STATE $w(\Ess_{k-1}) := \sum_{(P, S, l)\in \Ess_{k-1}}{|S| \choose l}$
    \STATE $\forall (u,v)\notin E(G): Z_{u,v} := 0$
    \FOR{$t := 1,\ldots, s$}
    \STATE $\forall j: X_{j,t} := 0$
    \STATE $\forall (u,v)\notin E(G): Z_{u,v,t} := 0$
    \STATE $H := \Sample(\Ess_{k-1})$ (Algorithm~\ref{alg:shadow-sample})
    \IF{$\exists j: H = K^{k-1}_j$ (H is a clique of size $k-1$)}
    \STATE $X_{j,t} := 1; \forall j'\neq j: X_{j',t} := 0$ %\text{Note: For the purpose of analysis only.}
    \FOR{$u \in N^*(H)$ (each neighbor of $H$)}
    \IF{$\exists v\in H: (u,v)\notin E(G)$ and $\forall w\in H, w\neq v: (u, w)\in E(G)$ and $u < v$}
    \STATE $Z_{u,v,t} := 1$
    \ENDIF
    \ENDFOR
    \ENDIF
    \STATE $\forall u<v, (u,v) \notin E(G): Z_{u,v} := Z_{u,v}+ Z_{u,v,t}$
    \ENDFOR
    \STATE $\forall u<v, (u,v) \notin E(G): \Konetil_{u,v} := Z_{u,v}\frac{w(\Ess_{k-1})}{s}$ 
    \STATE \textbf{Return }  $\Konetil$
  \end{algorithmic}

  \label{alg:pair-nearclique-approx}
\end{algorithm}

\begin{lemma}
  \label{lemma:K1uv-expectation}
  For every $u< v, (u,v)\notin E(G)$, $\E[Z_{u,v}\frac{w(\Ess_{k-1})}{s}] = \Kone_{u,v}$.
\end{lemma}

\begin{proof}
  By its definition, we have
  \begin{align*}
    Z_{u,v} =& \sum_{t:1\ldots s}Z_{u,v,t}\\
            =& \sum_{t:1\ldots s}
              \sum_{\substack{j:\exists i: \{u\} = V(K^{k-1}_i) \setminus V(K^{k-1}_j), \\ \{v\} = V(K^{k-1}_j) \setminus V(K^{k-1}_i)}} X_{j,t},
  \end{align*}
  %,  where the last inequality is because $Z_{u,v,t}=1 \iff \sum_{j:(u,v)\in E(K_j)}X_{j,t} = 1$, i.e., one of the clique containing $(u,v)$ is sampled.
  where the last equality is $Z_{u,v,t} = 1$ if and only if at time $t$, there exists a $k-1$-clique $K^{k-1}_j$ is sampled such that $v\in K^{k-1}_j$, and $u$ belongs to a clique $K^{k-1}_i$ (that is not sampled at $t$), such that they share all vertices, but mutually exclusive contain $u$ and $v$.
  In other words, $K^{k-1}_j \cup K^{k-1}_i$ forms a near-clique missing $(u,v)$. 
  It is clear that each unique pair of $K^{k-1}_j \cup K^{k-1}_i$ forms a unique near-clique missing $(u,v)$.
  Also, given a fixed $K^{k-1}_j$ sampled at step $t$, and a fixed $(u,v), u < v$, there exists at most one $K^{k-1}_i$ satisfying the condition, a one-to-one mapping between all pairs $i,j$ to all near-cliques missing $(u,v), u < v$.
  It is because each near-clique $K^{-1}$ missing $(u,v)$ corresponds to a unique unordered pair $K^{k-1}_j$  and $K^{k-1}_i$ that $u\in K^{k-1}_j$ and $v \in K^{k-1}_i$, as well as $v \notin K^{k-1}_j$ and $u \notin K^{k-1}_i$.

  Theorem~\ref{theorem:shadow-sample} shows that any clique $K^{k-1}_j$ can be sampled from $\Ess_{k-1}$~with probability $\frac{1}{w(\Ess_{k-1})}$ by Sample$(\Ess_{k-1})$, hence $X_{j,t}$ is a Bernoulli variable with $\Pr[X_{j,t} = 1] = \frac{1}{w(\Ess_{k-1})}$ and $\E[X_{j,t}] = \frac{1}{w(\Ess_{k-1})}$.
  By the linearity of expectations,

  \begin{align*}
    \E[Z_{u,v}] &= \sum_{t:1\ldots s} \sum_{\substack{j:\exists i: \{u\} = V(K^{k-1}_i) \setminus V(K^{k-1}_j), \\ \{v\} = V(K^{k-1}_j) \setminus V(K^{k-1}_i)}}\E[X_{j,t}] \\
                &= \sum_{t:1\ldots s}\sum_{\substack{j:\exists i: \{u\} = V(K^{k-1}_i) \setminus V(K^{k-1}_j),  \\ \{v\} = V(K^{k-1}_j) \setminus V(K^{k-1}_i)}}\frac{1}{w(\Ess_{k-1})} \\
                &= \frac{s}{w(\Ess_{k-1})}K^{-1}_{u,v},
  \end{align*}

where the last equality is because the set of $k-1$-clique pairs satisfying \[{j:\exists i: u = V(K^{k-1}_i) \setminus V(K^{k-1}_j), v = V(K^{k-1}_j) \setminus V(K^{k-1}_i)}\] also defines all the near-clique missing $(u,v)$, and the Lemma follows.
\end{proof}

\begin{lemma}
  \label{lemma:Koneuv-whp-bound}
  Given a constant $\theta$ (the approximation factor), and a failure probability $\delta$, if the number of samples $s \geq \frac{3w(\Ess_{k-1})\log{2/\delta}}{\theta^2\Kone_{u,v}}$, then with probability at least $1-\delta$, $Z_{u,v}\frac{w(\Ess_{k-1})}{s} \in (1\pm\theta)\Kone_{u,v}$, i.e., $\Konetil_{u,v} \in (1\pm\theta)\Kone_{u,v}$.
\end{lemma}

\begin{proof}
  For a fixed $t$, let $Y_{u,v,t} = \sum_{j:\exists i: u = V(K^{k-1}_i) \setminus V(K^{k-1}_j), v = V(K^{k-1}_j) \setminus V(K^{k-1}_i)}X_{j,t}$.
  It is clear to see that $Y_{u,v,t} = 1$ if and only if a near clique missing  $(u,v)$ is sampled at step $t$.
  Across different $t$, $Y_{u,v,t}$s are independent, with $\E[Y_{u,v,t}] = \E[\sum_{j:\exists i: u = V(K^{k-1}_i) \setminus V(K^{k-1}_j), v = V(K^{k-1}_j) \setminus V(K^{k-1}_i)}X_{j,t}] = \frac{\Kone_{u,v}}{w(\Ess_{k-1})}$.
  Since $Z_{u,v} = \sum_{t:1\ldots s}\sum_{j:\exists i: u = V(K^{k-1}_i) \setminus V(K^{k-1}_j), v = V(K^{k-1}_j) \setminus V(K^{k-1}_i)}X_{j,t}=\sum_{t:1\ldots s}Y_{u,v,t}$, with $Y_{u,v,t}$ being a Bernoulli variable with $E[Y_{u,v,t}] = \frac{\Kone_{u,v}}{w(\Ess_{k-'})}$, using the Chernoff bound, we have:
 \begin{align*}
   \Pr[Z_{u,v} \notin (1\pm\theta) \E[Z_{u,v}]] &\leq 2\exp(-\frac{\theta^2\E[Z_{u,v}]}{3})\text{, }\\
   \Pr[Z_{u,v} \notin (1\pm\theta) \Kone_{u,v}\frac{s}{w(\Ess_{k-1})}] &\overset{(a)}\leq 2\exp(-\frac{\theta^2\Kone_{u,v}{s}}{3{w(\Ess_{k-1})}}), \\
   \Pr[Z_{u,v}\frac{w(\Ess_{k-1})}{s} \notin (1\pm\theta) \Kone_{u,v}] &\overset{(b)}\leq 2\exp(-\log{2/\delta}), \\
    &\leq \delta, \\
 \end{align*} 
 where $(a)$ is because of Lemma~\ref{lemma:K1uv-expectation}, and $(b)$ is because of substituting $s$,
 and the Lemma follows.
\end{proof}

\begin{corollary}
 \label{cor:near-clique-point-wise-accurate}
 Let $J = \{(u,v):(u,v)\notin E(G)\}$, given a constant factor $\theta$, failure probability $\delta$, \NearCliqueAlg~is $(\theta,\delta)$-point-wise accurate.
\end{corollary}

After proving that $\NearCliqueAlg$ satisfies the properties of a point-wise statistic, we apply the routines stated in the Section~\ref{sec:max-point-wise}.

\begin{algorithm}[ht]
  \caption{\textsc{Max.PairWise.NearClique.Approx} \\ 
    \textbf{Input:} $G, k, \theta, \delta$ \\
    \textbf{Output:} An  approximation of $\max_{(u,v\notin E(G))} \Kone_{u,v}$}

  \begin{algorithmic}[1]
    \STATE $\Ess_{k-1} := \Turan(G, k-1)$
    \STATE $w(\Ess_{k-1}) := \sum_{(P,S,l)\in \Ess_{k-1}}{|S|\choose l}$
    \STATE $\tau_t := Algorithm~\ref{alg:tau-estimation}\big($ 
    
     \hspace{\algorithmicindent}
     $\NearCliqueAlg,$

     \hspace{\algorithmicindent}
     ${\max_{v\in V(G)}\degr{v} \choose k-2},3w(\Ess_{k-1}),\delta\big)$

    \STATE \textbf{Return} $Algorithm~\ref{alg:stat-estimation}\big($ 

     \hspace{\algorithmicindent}
     \NearCliqueAlg,

     \hspace{\algorithmicindent}$3w(\Ess_{k-1}), \theta, \delta, \tau_t\big)$
  \end{algorithmic}

  \label{alg:max-pair-nearclique-approx}
\end{algorithm}

\begin{theorem}
 \label{theorem:max-pair-nearclique-utility}
 With probability at least $1-4\delta$, \textsc{Max.PairWise.NearClique.Approx} returns

 $\max_{(u,v)\notin E(G)}\Konetil_{u,v} \in (1\pm\theta)\Konehat$.
\end{theorem}

\begin{proof}
  Corollary~\ref{cor:near-clique-point-wise-accurate} states that \NearCliqueAlg~is $(\theta,\delta)$-point-wise accurate regarding the estimation of $\Kone_{(u,v)}$, therefore we can apply the Point-wise statistic estimation in Section~\ref{sec:point-statistics}.
  Applying Lemma~\ref{lemma:binary-search-tau}, we have $\tau_t$ in Algorithm~\ref{alg:max-pair-nearclique-approx} satisfies that $\tau_t \leq \Konehat \leq 4\tau_t$ with probability at least $1-\delta$.
  Applying Lemma~\ref{lemma:stat-estimation}, with $\tau_t$ satisfying the requirements of the Algorithm~\ref{alg:stat-estimation}, the output of $Algorithm~\ref{alg:stat-estimation}$ with the first input parameter being $\NearCliqueAlg$ and the rest of the input being $3w(\Ess_{k-1}), \theta, \delta, \tau_t$ is $\in (1\pm\theta)\Konehat$ with probability at least $1-3\delta$.
  Taking the union bound on the failure probabilities of Algorithm~\ref{alg:tau-estimation} and Algorithm~\ref{alg:stat-estimation}, the Theorem follows.
\end{proof}

\textbf{Combining together.}
Algorithm~\textsc{Max.Clique.Pair} combines all these steps for estimating $LS_\kcl(G)$.

\begin{algorithm}[ht]
  \caption{\textsc{Max.Clique.Pair} \\ 
    \textbf{Input:} $G, h, \gamma, \delta_\kcl$ \\
    \textbf{Output:} An $(\gamma, \delta_\kcl )$-upper approximation of $LS_\kcl(G)$}

  \begin{algorithmic}[1]
    \STATE Let $\theta = \min(\frac{e^\gamma-1}{e^\gamma+1}, 1/2)$ 
    %\STATE Let $\delta = \delta_\kcl / 8$
\STATE Let $\delta = \delta_\kcl / 8$ \COMMENT{$8 = 4 + 4$, the failure-probability multipliers from Theorems~\ref{theorem:max-pair-clique-utility} and~\ref{theorem:max-pair-nearclique-utility}}

  \STATE $\widetilde{K}_{\max} := 1/(1-\theta) \times$
  \STATE \hspace{\algorithmicindent}
  $\textsc{Max.PairWise.Clique.Approx}(G,k,\theta,\delta)$
  \COMMENT{(Algorithm~\ref{alg:max-pair-clique-approx})}
  \STATE $\Konetil_{\max} := 1/(1-\theta) \times$
  \STATE \hspace{\algorithmicindent}
    $\textsc{Max.PairWise.NearClique.Approx}(G,k,\theta,\delta)$
  \COMMENT{(Algorithm~\ref{alg:max-pair-nearclique-approx})}
    \STATE \textbf{Return} $\max(\Ktil_{\max}, \Konetil_{\max})$
  \end{algorithmic}

  \label{alg:max-pair-approx}
\end{algorithm} 

\begin{restatable}[]{lemma}{maxcliqueapprox}
With probability at least $1-\delta_\kcl$, \textsc{Max.Clique.Pair} outputs a $(\gamma, \delta_\kcl)$-upper approximation of $LS_\kcl(G)$.
  \label{lemma:max-clique-approx}
\end{restatable}

\begin{proof}
  Theorem~\ref{theorem:max-pair-clique-utility} shows that Algorithm~\ref{alg:max-pair-clique-approx} outputs $\max_{u,v}\Ktil_{uv}\in (1\pm\theta)\Khat$ with probability at least $1-4\delta$.
  Theorem~\ref{theorem:max-pair-nearclique-utility} shows that Algorithm~\ref{alg:max-pair-nearclique-approx} outputs $\max_{u,v}\Konetil_{uv}\in (1\pm\theta)\Konehat$ with probability at least $1-4\delta$.
  By the union bound over the failures of Theorem~\ref{theorem:max-pair-clique-utility} and~\ref{theorem:max-pair-nearclique-utility}, both hold simultaneously with probability at least $1-4\delta - 4\delta = 1-8\delta = 1-\delta_\kcl$.
  Multiplying both with $1/(1-\theta)$ guarantee their estimations are larger than or equal to their true values, but not larger than $\frac{1+\theta}{1-\theta} \leq e^\gamma$ times their true values, and the Lemma follows.
\end{proof}

\begin{restatable}[]{theorem}{totaltimedirect}
\label{theorem:total-time-direct}
Let $\Ess_{k}, \Ess_{k-1}$ be the \Turan~computed from the input graph $G$ with clique size $k$ and $k-1$ respectively, and $w$ be a function defined in Algorithm~\ref{alg:max-pair-clique-approx}.
For small $\eps\lessapprox 1$,
the expected running time of
 $\tilde{S}_{\beta}(G)$ via the direct estimation of $\max_{u,v}f_{(k-2)\mathbb{C}}(G(A_{uv}))$ is: 
 \begin{align*}
O\Big(\frac{w(\Ess_{k})\log{m}\log^2{1/\delta}}{\eps^2\Khat} + \frac{w(\Ess_{k-1})\log{n}\log^2{1/\delta}}{\eps^2\Konehat} +\\ n\mathcal{D}(G)^{k} + m + n \Big).
 \end{align*}
\end{restatable}

\begin{proof}
Algorithm~\ref{alg:max-pair-clique-approx} invokes Algorithms~\ref{alg:tau-estimation} and~\ref{alg:stat-estimation} with $J = \{(u,v): (u,v)\in E(G)\}$, hence $|J| = m$.
By Theorem~\ref{theorem:max-stat-runtime}, it takes $O\left(\frac{3w(\Ess_{k})\log{m}\log^2{1/\delta}}{\eps^2\Khat}\right)$, as we set $\theta = \Theta(\eps/\log{1/\delta})$ for small $\eps\lessapprox 1$ ($\theta$ also depends on $\delta$, which varies due to the graph's size, hence the approximation), $w$ is as defined in Algorithm~\ref{alg:max-pair-clique-approx} and $\Ess_{k}$ Tur\'an's prefix shadow for cliques of size $k$ (see Section~\ref{sec:turan-shadow} for further analysis of the calculation of the Tur\'an's shadow).
Algorithm~\ref{alg:max-pair-nearclique-approx} invokes Algorithms~\ref{alg:tau-estimation} and~\ref{alg:stat-estimation} with $J = \{(u,v): (u,v)\notin E(G)\}$, hence $|J| = \binom{n}{2}-m$.
By Theorem~\ref{theorem:max-stat-runtime}, it takes $O\left(\frac{3w(\Ess_{k-1})\log{n}\log^2{1/\delta}}{\eps^2\Konehat}\right)$, again, as we set $\theta = \Theta(\eps/\log{1/\delta})$ for small $\eps\lessapprox 1$, and $\Ess_{k-1}$ as the Tur\'an's prefix shadow for cliques of size $k-1$.
In both Algorithms above, we have to take into account the time required for calculating the Tur\'an's shadows (for $k$-and $k-1$-cliques respectively), which takes extra $O(n\mathcal{D}(G)^{k} + m + n)$ and $O(n\mathcal{D}(G)^{k-1} + m + n)$, where $\mathcal{D}(G)$ is the degeneracy of $G$ (see Section~\ref{sec:turan-shadow}).
We note that with larger $\eps$, the factor $1/\eps^2$ vanishes due to $\theta = 1/2$.
\end{proof}

%%% Local Variables:
%%% mode: latex
%%% TeX-master: "main"
%%% End:

\section{Experimental Evaluation}
\label{sec:exprs}

\begin{table*}[tb]
    \centering
    \footnotesize
    \begin{tabular}{|c|c|c|c|c|c|c|}
    \hline
        Network & Description & \#nodes & \#edges & $f_{4\mathbb{C}}$ & $f_{5\mathbb{C}}$ & $f_{6\mathbb{C}}$ \\
        \hline
        \hline
         \texttt{ca-HepPh} & Collaboration network- Arxiv High Energy Physics & 12,008 & 237,010 & $1.502 \times 10^8$ & $6.491\times 10^9$ & $2.464 \times 10^{11}$ \\
         \texttt{ca-AstroPh} & Collaboration network- Arxiv Astro Physics & 18,772 & 198,110 & $9.580 \times10^6$ & $6.499 \times 10^7$ & $4.004\times 10^8$ \\
         \texttt{email-Enron} & Email communication network- Enron & 36,692 & 367,662 & $2.344 \times 10^6$ & $5.810\times 10^6$ & $1.121\times10^7$ \\
         \texttt{loc-Gowalla} & Gowalla location based online social network & 196,591 & 950,327 & $6.086 \times 10^6$  & $1.457 \times 10^7$  & $2.892 \times 10^7$\\   
         %\texttt{com-Amazon} & Amazon product co-purchasing network & 334,863 & 925,872 & 276,065 & 61,554  & 5798   \\
         \texttt{com-Youtube} & Youtube social media community network & 1,134,890 & 2,987,624 & $4.986 \times 10^6$ & $7.211\times 10^6$ & $8.443\times 10^6$ \\ 
         %\texttt{cit-Patents} & Citation network among US Patents & 3.774,768  & 16,518,947 & $3.501 \times 10^6 $ & $3.039 \times 10^6$ & $3.152 \times 10^6$   \\
         \hline
    \end{tabular}
    \caption{Statistics of networks including estimates of $k$-clique counts}
    \label{tab:networks}
\end{table*}

\subsection{Experimental setup and datasets}

% We compare our private k-clique counting algorithm which uses  approximations to the local sensitivity, $LS_{\kcl}(G)$, and largest count of common neighbors for any node pairs, $\hat{a}$, against the ladder function method by \cite{zhang:sigmod15} which uses exact computation of $LS_{\kcl}(G), \hat{a}$. 

\textbf{Datasets. } 
We consider various real-world networks (from the Stanford Large Network Dataset Collection~\cite{snapnets}) on which we compute the private $k$-clique counts. 
Statistics of these networks are shown in Table \ref{tab:networks}. 
The largest network has nearly 3 million edges, while the smallest has around 200k edges.

%which is an order of magnitude larger than the largest networks considered in all prior work \cite{nguyen2023faster}.

% The smallest network has about 198k edges while the largest has over 16.5 million edges. All networks have been made available on Stanford Large Network Dataset Collection~\cite{snapnets}. 

% \rmcomment{Add 2-3 more networks}

\textbf{Baseline}
We compare our algorithms against the ladder function method~\cite{zhang:sigmod15}, which uses exact computation of $LS_{\kcl}(G), \hat{a}$ by exhaustive search over all node pairs, with pruning. 
%\rmcomment{Add algorithm for this exact computation which can be moved to supplement. }

\textbf{Evaluation metrics. } We compare our method with the baseline in terms of accuracy and runtime performance. \\
1. Accuracy: We use relative error defined as $\frac{|M_{f_{\kcl}}(G) - f_{\kcl}(G)|}{f_{\kcl}(G)}$. \\
2. Speedup: We compute the speedup which compares how much faster our algorithm is over the baseline, i.e., $\frac{\text{Runtime of ladder function baseline}}{\text{Runtime of our algorithm on G}}$.

We test our methods at different privacy budgets $\epsilon \in \{2^{-2,-1,0,1,2}\}$ and a fixed $\ddp=10^{-5}$, and clique sizes $k\in \{4,5,6\}$. We report the average relative error and runtime over five repeated runs of our algorithms due to the randomness in the estimation of $LS_{\kcl}(G), \hat{a}$. For the ladder baseline, we exactly calculate the local sensitivity and maximum neighborhood size once. Then we resample the noise from the ladder function in each run to obtain five samples. We take the average accuracy and runtime of both methods in the reported values.

\textbf{Computing infrastructure. } All algorithms are implemented in C++. We ran our experiments on a system with Intel Xeon Gold 6148 CPU and 375GB of RAM with 40 physical CPU cores.

% \textbf{Environment setup}

% \textbf{Real-world networks.}

% \textbf{Synthetic networks.}

\subsection{Experimental results} We report accuracy values for which the exact $k$-clique counts are either known from previous work~(\cite{finocchi2015clique}) or could be computed within 48 hours. In many settings, the ladder baseline with exact computation did not complete within 48 hours.

\textbf{Impact of clique size $k$ on accuracy and speedup.} Figure~\ref{fig:rel_err} shows that \textsc{FastCliqueDP} achieves similar levels of accuracy to the ladder baseline across various clique sizes $k$: 
the relative error of our algorithm remains within $1\%$ on the smaller  networks, \texttt{ca-HepPh}, \texttt{ca-AstroPh}, \texttt{email-Enron}, even for $f_{\kcl}$ in the order of $10^{11}$. On \texttt{loc-Gowalla}, \textsc{FastCliqueDP} is far more accurate as compared to the ladder function especially with increasing $k$, for $\epsilon=4$. 
In Figure~\ref{fig:speedup}, we report the speedup over ladder baseline when using our method for different clique sizes $k$ and privacy budgets $\epsilon$. 
Our algorithm is orders of magnitude faster than the ladder baseline under many settings especially with increasing clique size $k$. 
We note that for the largest networks, the ladder function was not able to complete even in 48 hours.
\textsc{FastApproxCliqueDP} has even better run time than \textsc{FastCliqueDP}, but the accuracy is not very high as can be seen in Figure~\ref{fig:approx}.

%For $k>=6$, our approximation algorithm is at least an order of magnitude faster than the exact computation. This is because exact computation of $(k-2)$-clique counts in numerous neighborhoods for $LS_{\kcl}(G)$ take exponentially longer going from $k=5$ to $6$.
%This can been seen in the $30\times$ speedup on \texttt{ca-HepPh} for $k=6$.
% The ladder method's exact computation did not complete even after 48 hours for $k=\{7,8\}$, while our algorithm took under 3 hours in most runs~(Figure \ref{fig:all_nets}), indicating at least $10\times$ speedup. 

%Even on networks where the $k$-clique count increases exponentially with increasing $k$, our approximate algorithm grows much more slowly as seen in Figure~\ref{fig:all_nets}~\rmcomment{add concrete numbers}.

\textbf{Speedup and accuracy with $\epsilon$. } 
%\rmcomment{Add speedup vs. $\epsilon$ plots.}
Figure~\ref{fig:rel_eps} shows the relative error of both methods with varying $\epsilon$. 
\textsc{FastCliqueDP} consistently achieves similar or lower levels of error compared to the ladder baseline. With increasing $\epsilon$, the relative error for both methods decreases as expected. 
The speedup of \textsc{FastCliqueDP} increases rapidly with $\epsilon$ to multiple orders of magnitude, before plateauing.  
% In Figure~\ref{fig:speedup}, we show the speedup over ladder achieved by our method for varying $\epsilon$. With increasing budget, the gains in speedup increases rapidly to multiple orders of magnitude before plateauing. 
This is due to our method's approximate computations running much faster, while the ladder's exact computation has a fixed runtime irrespective of privacy budget. With our method, we can speedup the computation by allowing more privacy budget.

\textbf{Network size. } Our method scales to networks with millions of edges while still maintaining low relative error in most settings.  Figure~\ref{fig:runtime} shows that the runtime of our method increases gradually even as network sizes grow by orders of magnitude. In contrast, the ladder baseline did not complete on \texttt{com-Youtube} within 48 hours.
On large networks ($>=100k$ nodes), our method is many times faster even for small $k = \{4, 5\}$, due to the potentially large number of node-pairs that must be exhaustively searched in the exact computation. For example, on \texttt{com-Youtube}, $\epsilon = 0.5$,  our algorithm took between 0.5 hours ($k=4$) to 12 hours ($k=5$), while, in the ladder's exact computation, only 30\% of the node pairs could be searched in 48 hours, suggesting at least an order of magnitude speedup.

% \textbf{Using approximation for $f_{\kcl}(G)$} 
% \rmcomment{Suffices to report accuracy w.r.t our Alg 1 using exact $f_{\kcl}(G)$. Unlike previous sections, for speedup, compare total runtime including exact $f_{\kcl}(G)$ computation, w.r.t our Alg 1. Report these total runtimes on several  large networks and $k$ showing many orders of magnitude speedup.  }

\begin{figure}
    \centering
    \begin{subfigure}{0.48\columnwidth}
    \centering
    \includegraphics[width=\linewidth]{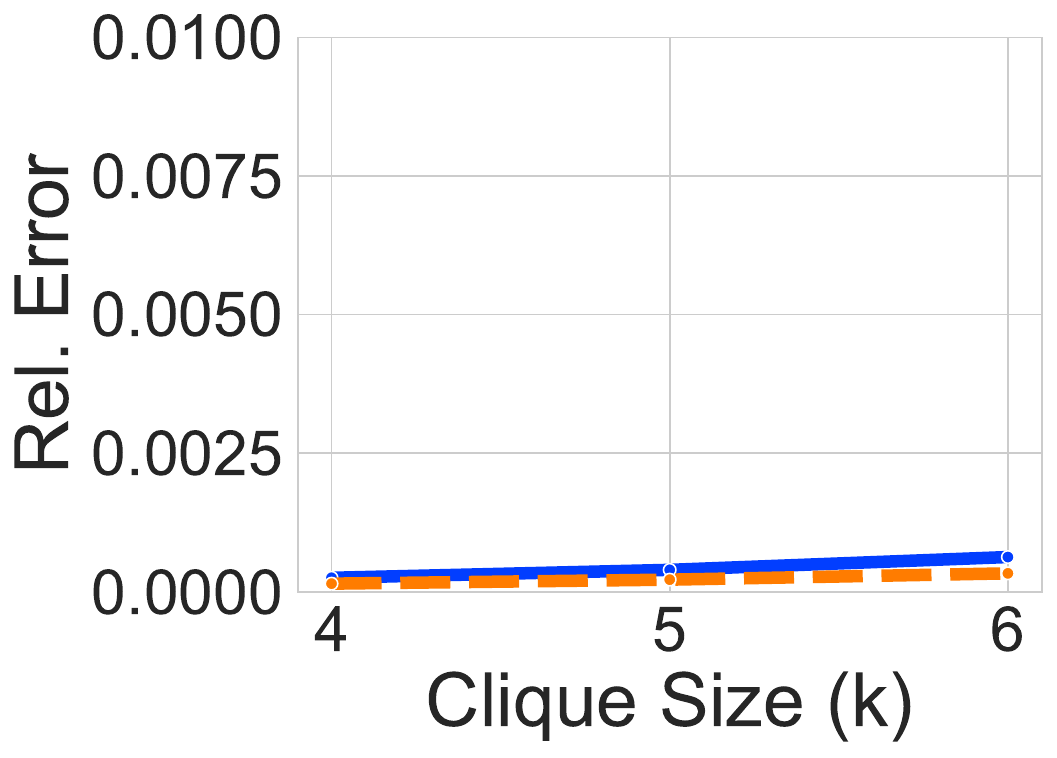}
    \caption{\texttt{ca-HepPh}}
    \end{subfigure}
    \begin{subfigure}{0.48\columnwidth}
        \centering
        \includegraphics[width=\linewidth]{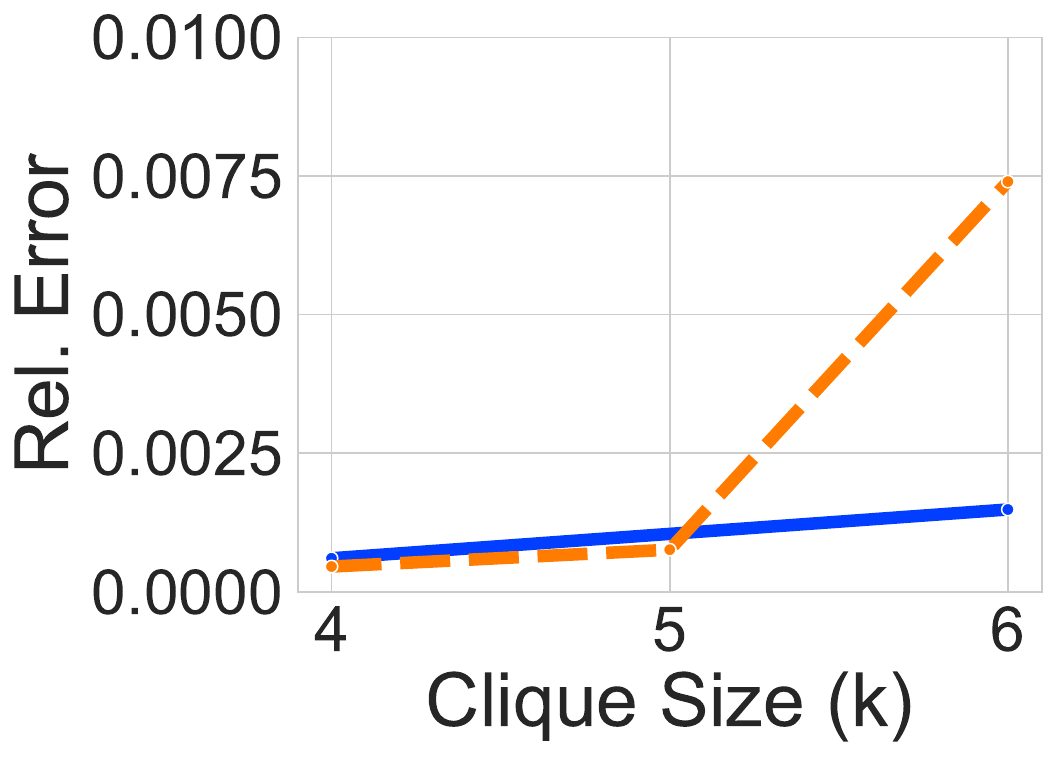}
    \caption{\texttt{ca-AstroPh}
    }
    \end{subfigure}

    \begin{subfigure}{0.48\columnwidth}
        \centering
        \includegraphics[width=\linewidth]{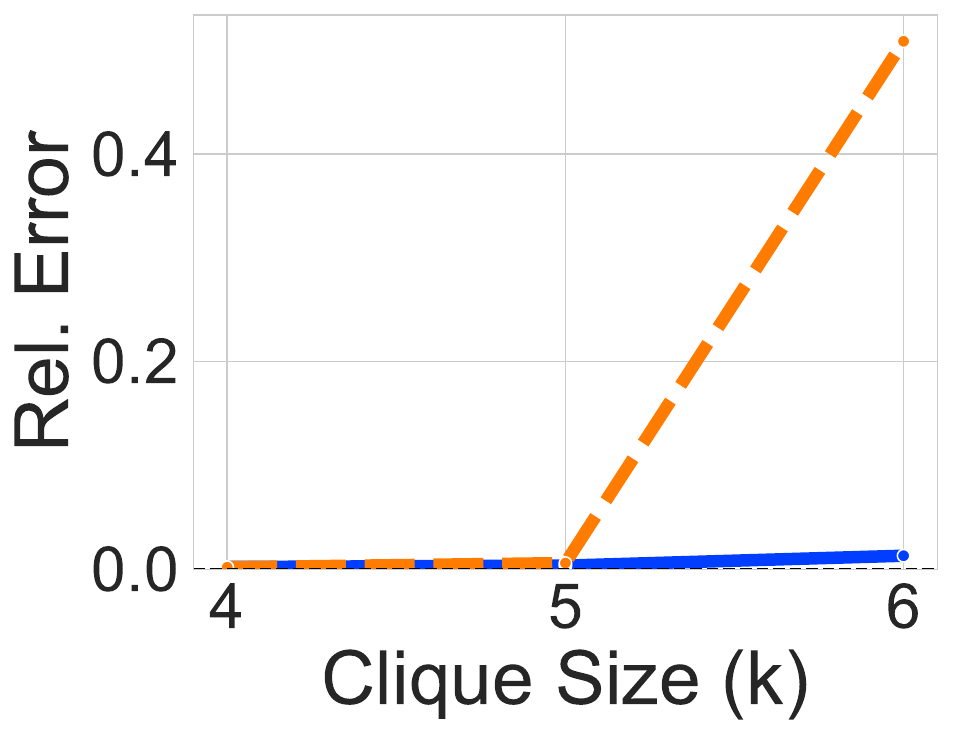}
    \caption{\texttt{email-Enron}
    }
    \end{subfigure}
    \begin{subfigure}{0.48\columnwidth}
        \centering
        \includegraphics[width=\linewidth]{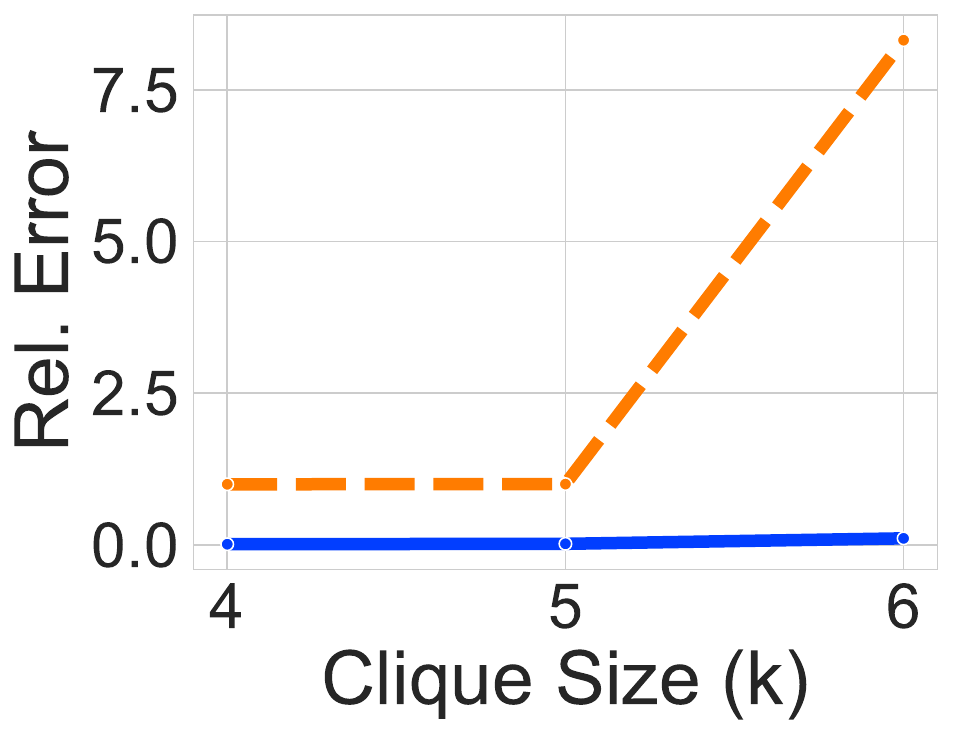}
    \caption{\texttt{loc-Gowalla}
    }
    \end{subfigure}

    \includegraphics[width=0.6\columnwidth]{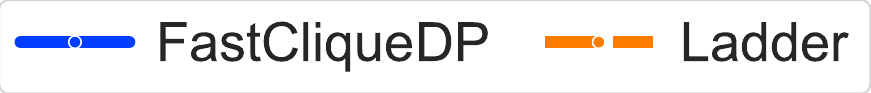}
    \caption{Relative Error for different clique sizes $k$ at $\epsilon=4$ on various networks.}
    \label{fig:rel_err}
\end{figure}
\begin{figure}
    \centering
    % \begin{subfigure}{0.48\columnwidth}
    %     \centering
    %     \includegraphics[width=\linewidth]{figures/ca-HepPh_speedup_vs_k_eps0.5.pdf}
    %     \caption{Varying $k$}
    % \end{subfigure}
    \begin{subfigure}{0.48\columnwidth}
        \centering
        \includegraphics[width=\linewidth]{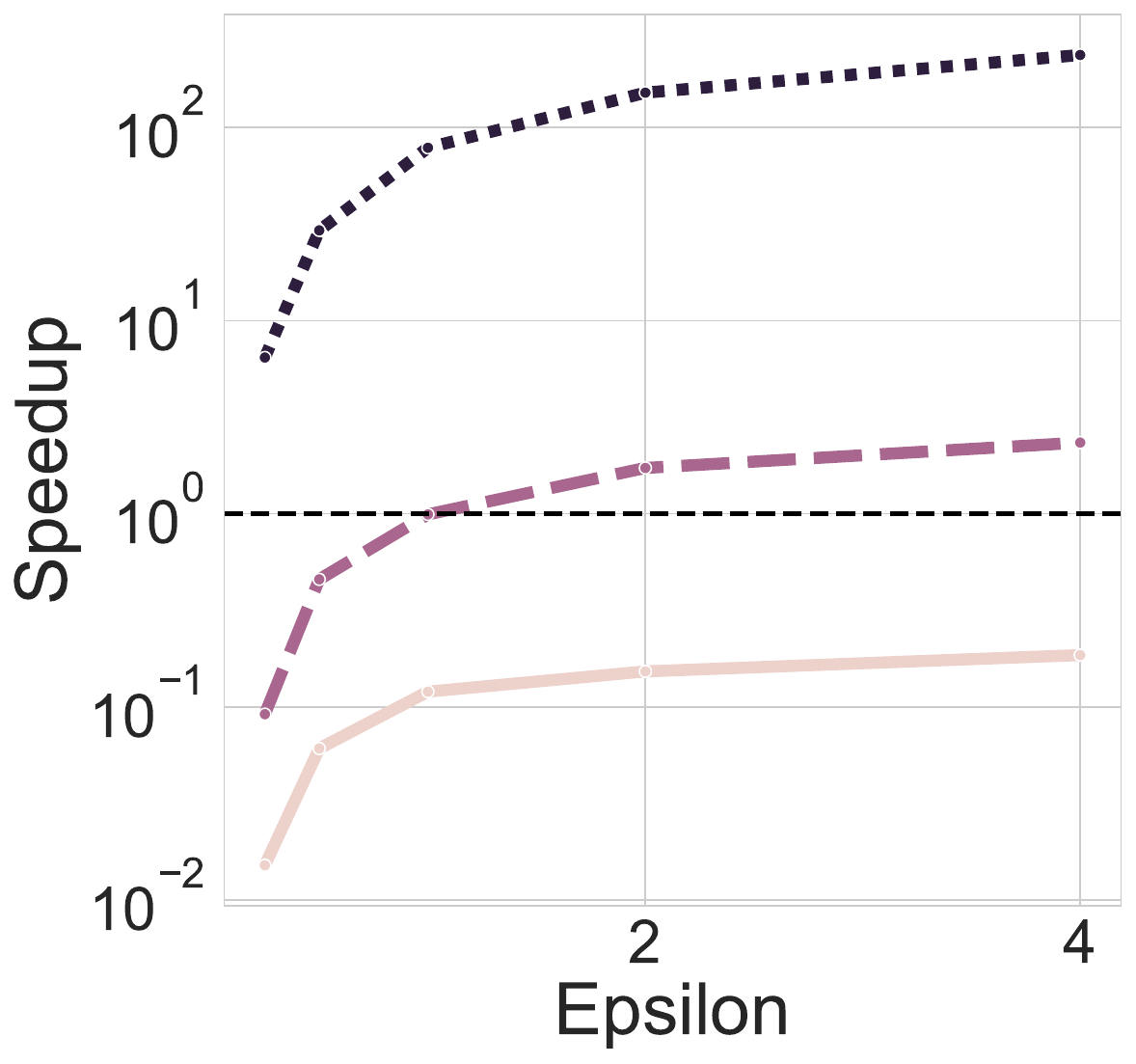}
        \caption{\texttt{ca-HepPh}}
    \end{subfigure}
    \begin{subfigure}{0.48\columnwidth}
        \centering
        \includegraphics[width=\linewidth]{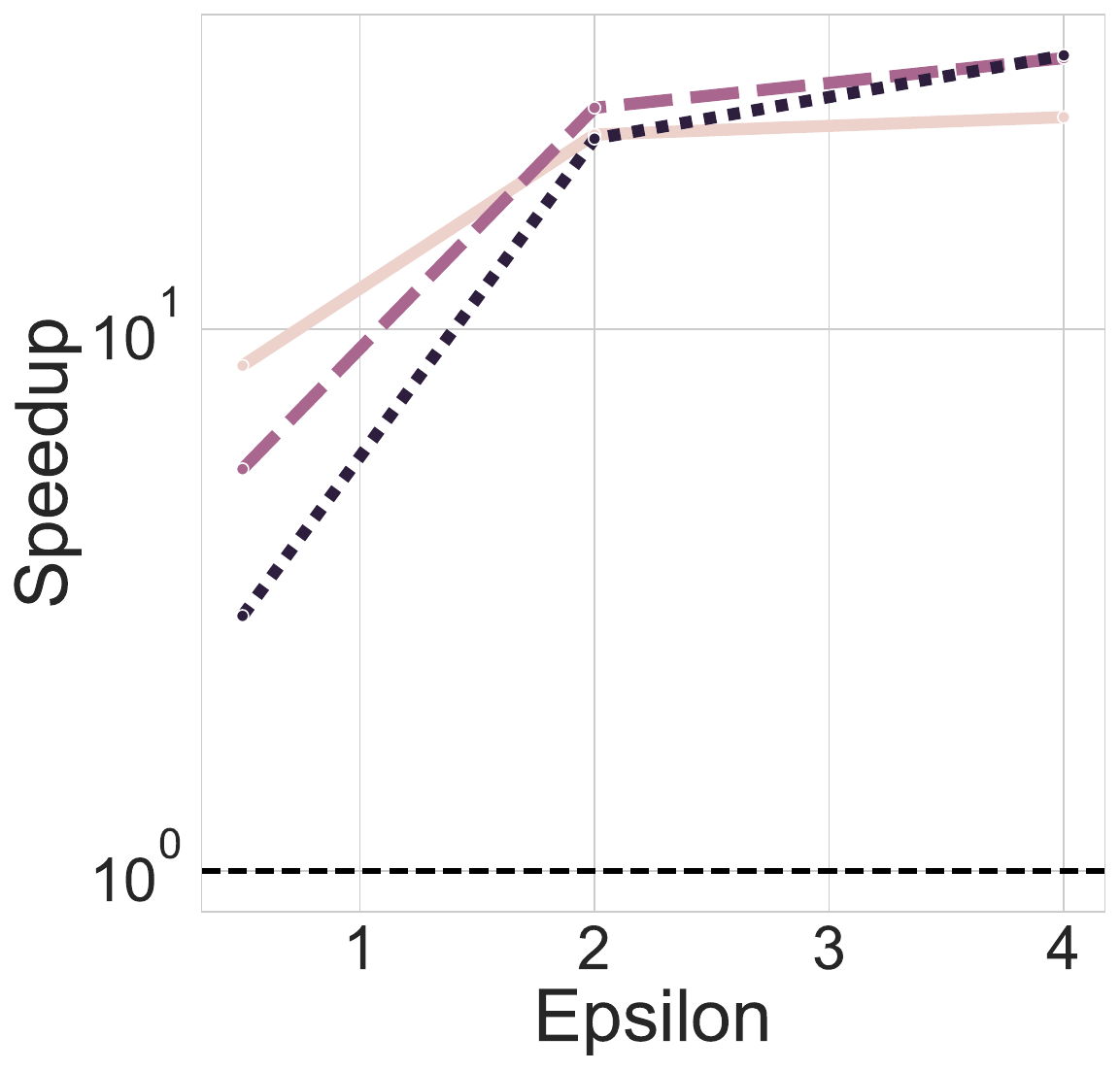}
        \caption{\texttt{loc-Gowalla}}
    \end{subfigure}
    \includegraphics[width=0.4\columnwidth]{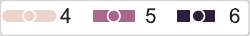}
    \caption{Speedup over ladder baseline for clique sizes $k\in\{4, 5, 6\}$ with varying $\epsilon$. Y-axis is in log-scale.}
    \label{fig:speedup}
\end{figure}

\begin{figure}
    \centering
    \includegraphics[width=0.5\columnwidth]{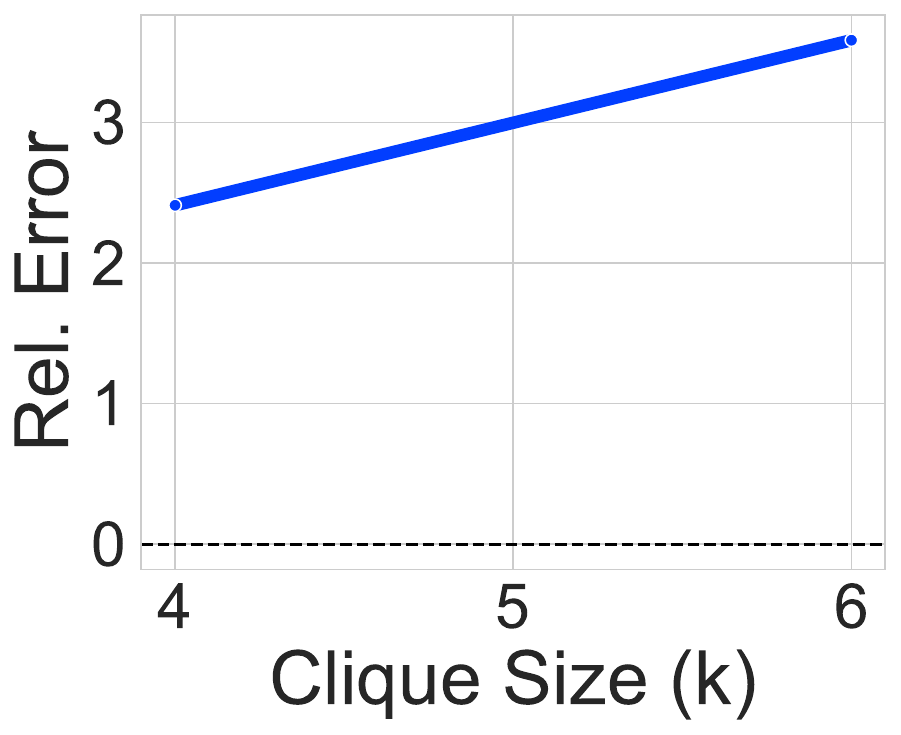}
    \caption{Relative Error in 
    \textsc{FastApproxCliqueDP} (Alg.~\ref{alg:clique-approx-dp-fast-approx}) for different clique sizes on \texttt{ca-HepPh} at $\epsilon = 12$.}
    \label{fig:approx}
\end{figure}

\begin{figure}
    \centering
    \begin{subfigure} {0.48\columnwidth}
    \centering
    \includegraphics[width=\linewidth]{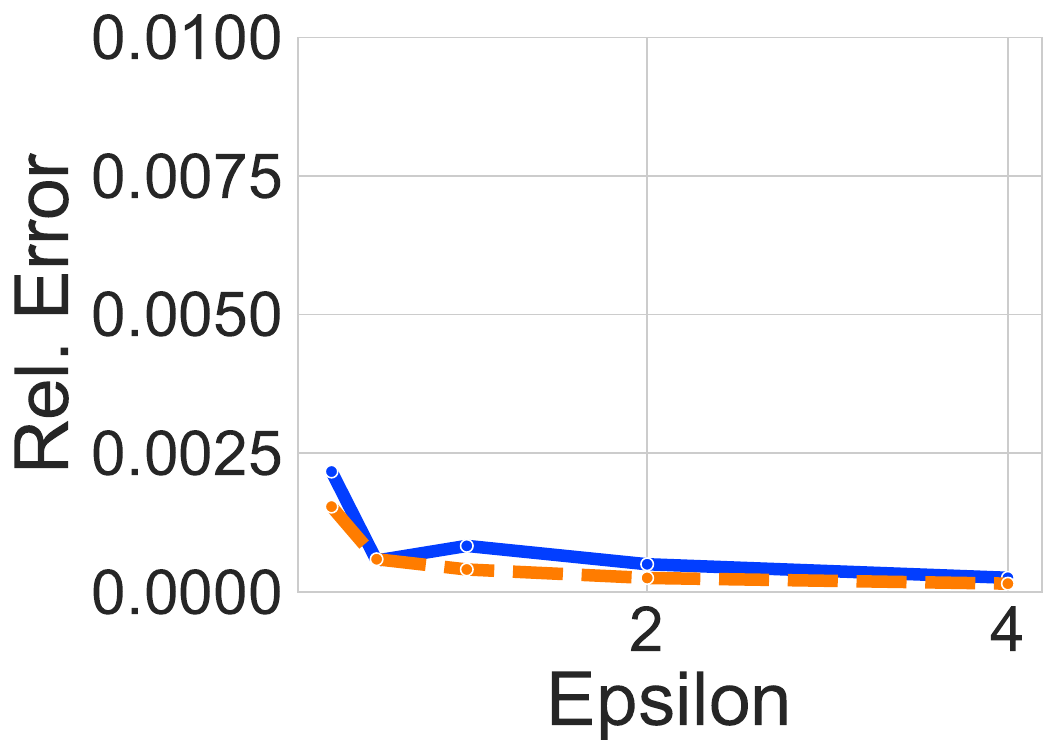}
        \caption{\texttt{ca-HepPh}: $k=4$}
    \end{subfigure}
    % \begin{subfigure} {0.48\columnwidth}
    % \centering
    % \includegraphics[width=\linewidth]{figures/ca-HepPh_k5_relerr_vs_eps.pdf}
    %     \caption{$k=5$}
    % \end{subfigure}
    \begin{subfigure}{0.48\columnwidth}
    \centering
    \includegraphics[width=\linewidth]{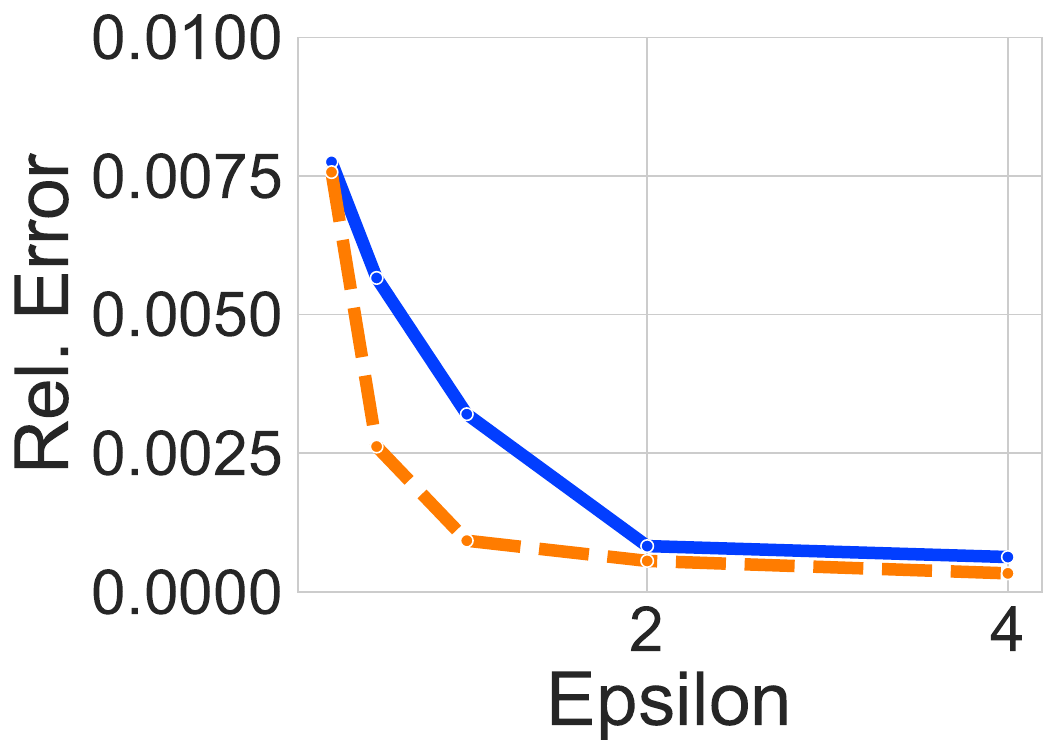}
        \caption{\texttt{ca-HepPh}: $k=6$}
    \end{subfigure}
    \begin{subfigure}{0.48\columnwidth}
    \centering
    \includegraphics[width=\linewidth]{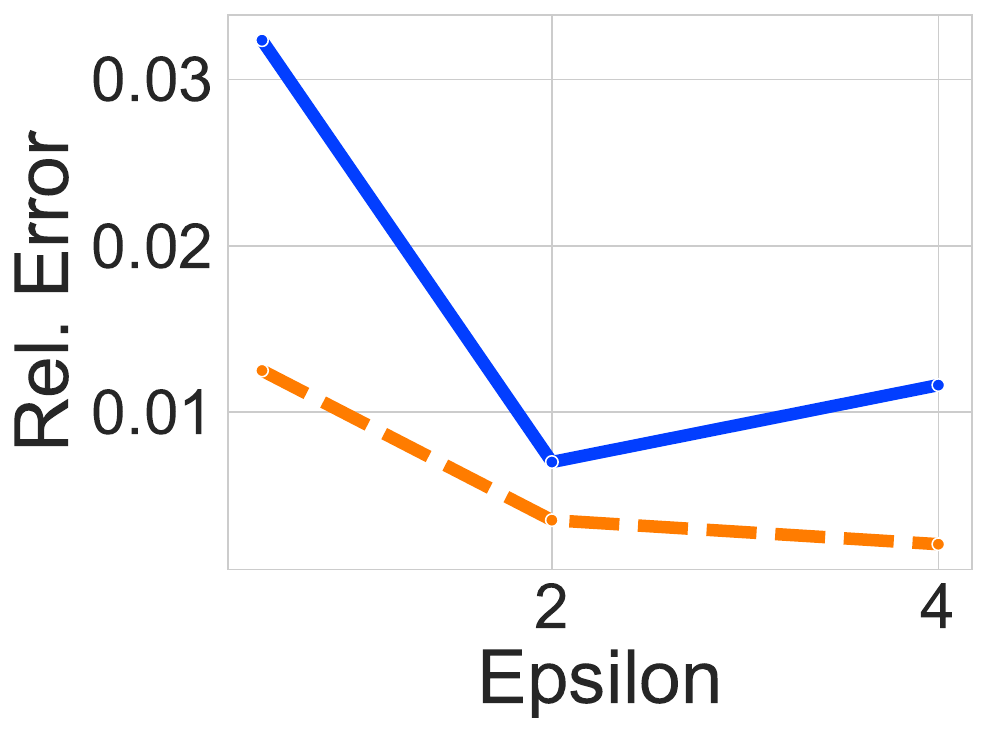}
        \caption{\texttt{loc-Gowalla}: $k=4$}
    \end{subfigure}
    \begin{subfigure}{0.48\columnwidth}
    \centering
    \includegraphics[width=\linewidth]{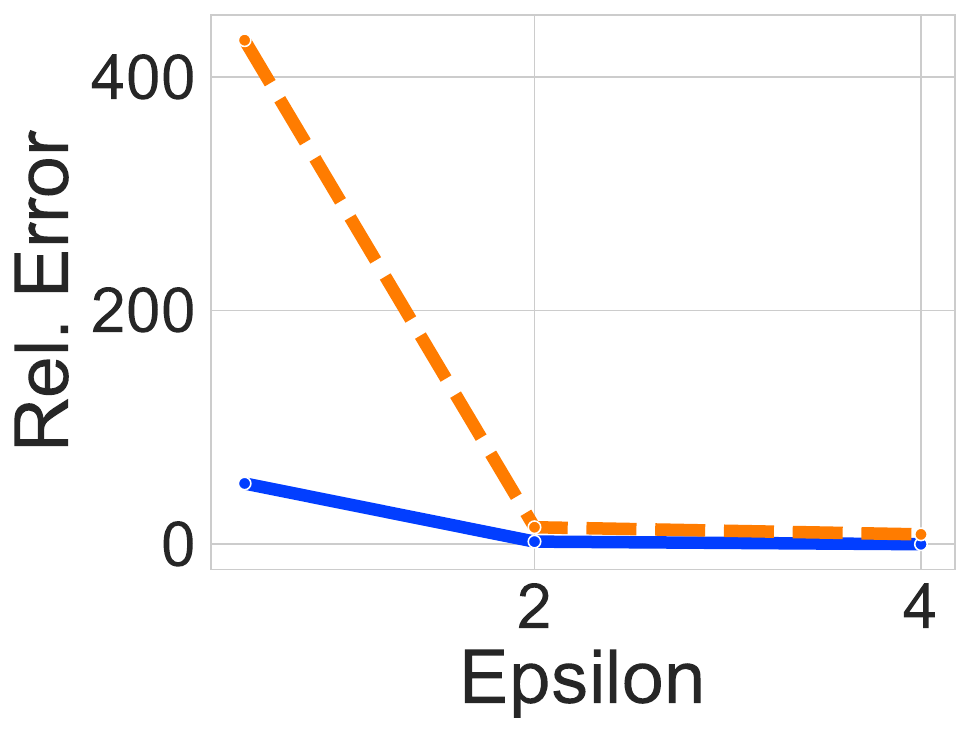}
        \caption{\texttt{loc-Gowalla}: $k=6$}
    \end{subfigure}
        \includegraphics[width=0.6\columnwidth]{figures/dp_methods_legend.pdf}
    \caption{Relative Error with varying privacy parameter $\epsilon$ comparing \textsc{FastCliqueDP} with ladder baseline.}
    
    \label{fig:rel_eps}
\end{figure}

\begin{figure}
    \centering
    \includegraphics[width=0.95\columnwidth]{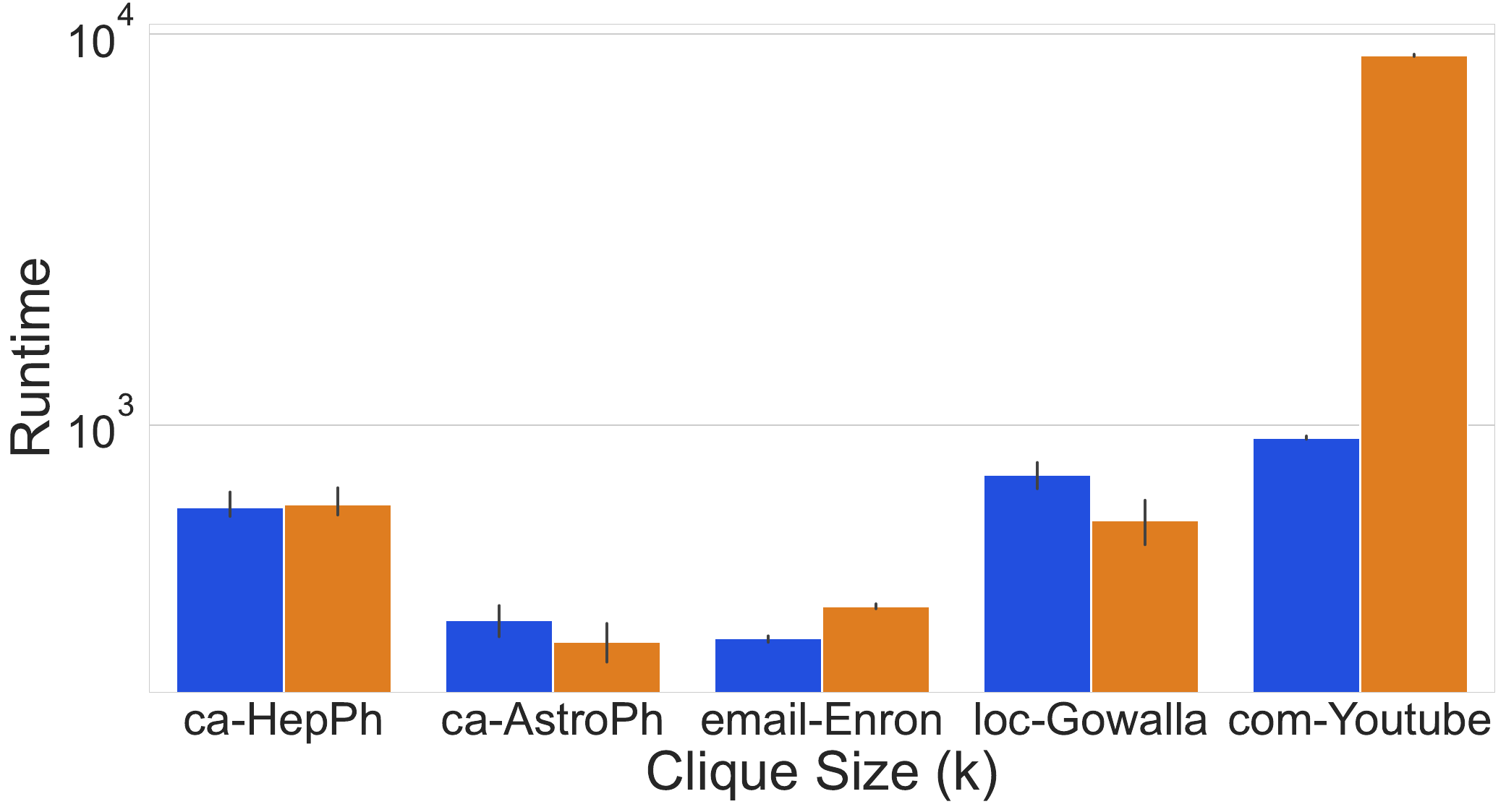}
    \includegraphics[width=0.4\columnwidth]{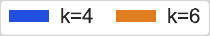}
    \caption{Runtime (in seconds) of \textsc{FastCliqueDP} on various networks at $\epsilon =4$.}
    \label{fig:runtime}
\end{figure}

% \begin{figure}
%     \centering
%     \begin{subfigure}{0.5\columnwidth}
%     \centering
%     \includegraphics[width=\linewidth]{figures/networks_time_vs_k_e0.5.pdf}
%         \caption{Runtime}
%     \end{subfigure}
    
%     \caption{All networks Rel. Error and Runtimes for $(\epsilon=0.5, \ddp=10^{-5})$.}
%     \label{fig:all_nets}
% \end{figure}

%%% Local Variables:
%%% mode: latex
%%% TeX-master: "main"
%%% End:

\section{Conclusions}

% Alternatives to global sensitivity of $k$-clique count, such as smooth sensitivity and ladder function, are all hard problems.
% Specifically, they all involve (exact) counting the number of cliques of slightly smaller size, such as the count of $(k-2)$-cliques.
% We modified the ladder function by adding specific constraints to create a smooth upper bound on the local sensitivity of the $k$-clique count.
% However, the smooth upper bound still includes the count of smaller-sized cliques.
% We utilized a fast algorithm to approximate the $(k-2)$-clique statistic~\cite{eden2018approximating}.
% The original algorithm, although it achieves any-constant approximation, has a constant failure probability, which is insufficient for privacy.
% We therefore invoked multiple instances of it, obtaining a median estimation, to get an inverse polynomial failure probability needed by the framework.
% As a result, we obtained a polynomial-time algorithm to estimate the smooth upper bound on the local sensitivity, which we can extend our framework to utilize it to calibrate the noise needed for privacy.

Graph statistics and other kinds of combinatorial problems have high sensitivity, and alternatives to global sensitivity based metrics are needed to improve accuracy.
However, most private algorithms involving alternatives to global sensitivity do not scale to large instances, which limits their use on real-world instances.
Our paper shows a novel approach, \textsc{FastCliqueDP}, to scale these techniques using approximations.
Such techniques have been crucial for scalability in graph mining, but there has been very limited work on incorporating them into privacy.
We show rigorous bounds on the performance of \textsc{FastCliqueDP}, which are matched by experimental evaluation on networks with millions of edges.
However, the exact type of approximation is important:
the accuracy of \textsc{FastApproxCliqueDP}, which uses approximate count instead of exact, is worse.
Improving this, and extending to other graph problems is interesting future work.

%%% Local Variables:
%%% mode: latex
%%% TeX-master: "main"
%%% End:

\clearpage

\bibliographystyle{plain}
\bibliography{refs.bib}

\appendix

% \section{Limitations of prior methods}
% \label{sec:limitations}
% 
% \begin{observation}
% $GS_{\fkcl} = \Theta(n^{k-2})$. 
% \end{observation}
%     
% \color{blue}
% For many networks $GS_{\fkcl}$ is larger than $\fkcl(G)$ 
% \color{black}

% \section{Approximate smooth sensitivity for approximate queries}
% 
% \begin{algorithm}
% \caption{DP estimate via approximate smooth sensitivity (corrected, tightened)}
% \label{alg:corrected}
% \textbf{Input:} dataset $D$; query $f$ with oracles $A_f$ (an $(\alpha_1,\delta_1)$
% multiplicative approximation) and $\tilde S$ (a $(\gamma,\delta_2)$ upper approximation
% to $S_{f,\beta}$); privacy budget $(\epsilon,\delta)$ with $\delta\le 2/e$. \\
% \textbf{Precondition:} $\alpha_1\in(0,1/2)$ and
% $\;4\alpha_1+\gamma+\beta\le \dfrac{\epsilon}{2\ln(2/\delta)}$. \\
% \textbf{Output:} $\bigl(\epsilon,\tfrac{e^{\epsilon/2}+1}{2}\delta+2\delta_1+2\delta_2\bigr)$-DP
% estimate of $f(D)$.
% \begin{algorithmic}[1]
%   \STATE $\hat a \leftarrow A_f(D)$ \hfill // multiplicative approx.\ of $f(D)$
%   \STATE $\tilde s \leftarrow \tilde S(D)$ \hfill // upper approx.\ of $S_{f,\beta}(D)$
%   \STATE $b \leftarrow \dfrac{1}{\epsilon}\left(\dfrac{4\alpha_1}{1-\alpha_1}\,\hat a + 4\,\tilde s\right)$ \hfill // Laplace scale
%   \STATE Sample $Z\sim \mathrm{Lap}(1)$
%   \STATE \textbf{Return} $\hat a + b\,Z$
% \end{algorithmic}
% \end{algorithm}

\subsection{Ladder function of $\fkcl$}
\label{sec:ladder}
In this section we present the Ladder function of $\fkcl$ and its property.
For the full details, we refer readers to~\cite{zhang:sigmod15}.

\begin{theorem}(Theorem 5.1 of~\cite{zhang:sigmod15})
  \label{theorem:clique-ladder}
  Let $LS_\kcl(G)$ be the local sensitivity of $\fkcl$ at $G$.
  Let $a_{u,v}$ be the number of common neighbors of node $u$ and $v$ in $G$.
  Let $\ahat = \max_{u,v}a_{u,v}$.
  The function
  \begin{align*}
    I_t(G) = \min\left(LS_\kcl(G) + \binom{\ahat + t}{k-2} - \binom{\ahat}{k-2}, GS_\kcl \right)
  \end{align*}
  is a ladder function for $\fkcl$.
\end{theorem}

%\subsection{Proofs from Section \ref{sec:approach}}

% \smoothclique*

% \noindent
% \textbf{Proof of Lemma \ref{lemma:I-approx}}

% \approxsmooth*
%\textbf{Proof of Theorem \ref{theorem:tildeS}}

% \privacy*
%\textbf{Proof of Theorem \ref{theorem:privacy}}

\subsection{Black-box approximation of  $LS_\kcl(G) =\max_{u,v}f_{(k-2)\mathbb{C}}(G(A_{uv})) $}
\label{sec:blackbox-appendix}

By its definition $LS_\kcl(G) = \max_{G': G'\sim G}|\fkcl(G') - \fkcl(G)|$, it indicates the largest difference in $\#k$-cliques of a neighbor $G'$ from $G$, in which $G$ and $G'$ differ in exactly one edge.
Assume that edge $(u,v)$ exists, a specific $k$-clique that involves $(u,v)$ contains $k-2$ nodes that are common neighbors of $u$ and $v$, and that they are all connected to each other, hence they form a $k-2$-clique if we remove $u$ and $v$.
Therefore, one way to calculate the number of $k$-clique involving $u,v$ is to count all $k-2$-cliques in a subgraph induced by the common neighbors of $u$ and $v$.
We can utilize any clique counting (with clique size $k-2$) algorithm for each induced subgraph, hence the approach is named ``black-box''.
For this purpose, we use \algK~\cite{eden2018approximating}, that claims to estimate $k$-clique in sublinear time (in specific regimes).
However, \algK~has failure probability of $1/3$, where our $\LStil(G)$ requires a failure probability much lower $O(1/n)$ as it will be incorporated into the $\delta$-part of the privacy.
We then have to sample $O(\log{n})$ instances of $\algK$, and takes the median in order to achieve the desired failure probability.

\begin{definition} Let $\widetilde{LS}(G) = \max_{u,v}\textsc{Clique.Approx}(G(A_{uv}), k-2, \gamma, \delta_{\kcl})$.
  \label{def:approx-ls-blackbox}
\end{definition}

\begin{lemma}Let \algK $(G, \alpha)$ be the algorithm defined in Corollary 20 of~\cite{eden2018approximating}, that estimates the number of $k$-clique in a graph $G$ with approximation factor $\alpha$.~\cite{eden2018approximating} shows that \algK$(G,\alpha) \in (1 \pm \alpha)f_\kcl(G)$ with probability at least $2/3$.
  \label{lemma:k-clique-alg}
\end{lemma}

\begin{lemma}Let $X$ be a random variable that $X\in(l,r)$ with probability $p>0.5$. Let $Y$ be the median of $x_1,\ldots, x_{2d}$ i.i.d. from $X$. With probability at least $1-\exp(-p(1-\frac{1}{2p})^2d)$, $Y\in(l,r)$.
  \label{lemma:median}
\end{lemma}

\begin{proof}
  We observe that if there exists at least $d$ instances of $x_i$ in $x_1,\ldots, x_{2d}$ lie between $l$ and $r$, then the median of $x_1,\ldots, x_{2d}$ also in the range $(l,r)$. 
  Let $B_i$ be a Bernoulli variable that $B_i = 1$ if and only if $x_i\in(l,r)$. 
  The number of instances $x_i$ that lie in the range $(l,r)$ can be represented by $Z = \sum_{i=1}^{2d}B_i$, or in other words, $Z \sim Binom(2d, p)$. 
  We then have $\Pr[Y \notin (l,r)] \leq \Pr[Z < d]$. 
  By using Chernoff bound, we have: $\Pr[Z < (1-\alpha)2dp] \leq \exp(-\alpha^2dp)$. 
  Setting $\alpha = 1- \frac{1}{2p}$ yields $(1-\alpha)2dp = (1/(2p))2pd = d$, or in other words, we have $\Pr[Z < d] \leq \exp(-p(1-\frac{1}{2p})^2d)$, and the Lemma follows.
\end{proof}

\begin{algorithm}[ht]
  \caption{\textsc{Clique.Approx} \\ 
    \textbf{Input:} $G,k, \gamma, \delta_\kcl$ \\
    \textbf{Output:} An $(\gamma, n^{-2}\delta_\kcl )$-upper approximation of $f_\kcl$}

  \begin{algorithmic}[1]
    \STATE Let $\alpha = \frac{e^\gamma-1}{e^\gamma+1}$
    \STATE $d = 24(2\log{n} - \log{\delta_\kcl})$ 
    \FOR {$i = 1\ldots 2d$}
    \STATE $x_i = \text{\algK}(G, k, \alpha)$ 
    \ENDFOR
    \STATE $Y = median(x_1,\ldots, x_{2d})$
    \STATE \textbf{Return} $\frac{Y}{1-\alpha}$
  \end{algorithmic}

  \label{alg:clique-approx}
\end{algorithm}

\begin{lemma} With probability at least $1-n^{-2}\delta_\kcl$, $\textsc{Clique.Approx}(G, \gamma,\delta_{\kcl})$ outputs an $(\gamma, n^{-2}\delta_\kcl)$-upper approximation of $f_\kcl(G)$.
  \label{lemma:k-approx}
\end{lemma}

\begin{proof}
Due to Lemma~\ref{lemma:k-clique-alg}, each $x_i\in(1\pm\alpha) f_\kcl(G)$ with probability at least $2/3$.
Applying Lemma~\ref{lemma:median} for all $x_1, \ldots, x_{2d}$, their median $Y \in (1\pm\alpha)f_{\kcl}(G)$ with probability at least $1-\exp(-p(1-\frac{1}{2p})^2d)$. 
Substituting $p = \frac{2}{3}, d = 24(2\log{n} - \log{\delta_\kcl})$, the failure probability is $\exp(-p(1-\frac{1}{2p})^2d) = \exp(-2/3(1-\frac{1}{4/3})^2d) = \exp(-1/24 \cdot 24(2\log{n} - \log{\delta_\kcl})) = \exp(-2\log{n} + \log{\delta_\kcl})= n^{-2}{\delta_\kcl}$ and the Lemma follows.
\end{proof}

In the rest of the Section, we calculate and analyze all other components of the smooth upper bound on the local sensitivity of $\fkcl$ using the definition of $\widetilde{LS}(G)$ as in Definition~\ref{def:approx-ls-blackbox}.

% \begin{lemma}
%   $\widetilde{LS}(G)$ is a $(\gamma,\delta_{\kcl})$-upper approximation of $LS_\kcl(G)$.
%   \label{lemma:local-approx}
% \end{lemma}

\localapprox*
\noindent
%\textbf{Proof of Lemma \ref{lemma:local-approx}}

\begin{proof}
Fix $u, v$, because of Lemma~\ref{lemma:k-approx}, we have $f_{(k-2)\mathbb{C}}(G(A_{uv})) \leq \textsc{Clique.Approx}(G(A_{uv}),\ldots) \leq e^{\gamma}f_{(k-2)\mathbb{C}}(G(A_{uv}))$ with probability $1-n^{-2}\delta_\kcl$. Taking the union bound on all pair $u,v$, there are $\binom{n}{2} < n^2 $ such pairs, each $f_{(k-2)\mathbb{C}}(.)$ fails with probability $n^{-2}\delta_\kcl$, $LS_\kcl(G) \leq \widetilde{LS}(G) \leq e^{\gamma}LS_\kcl(G)$ with probability $1-\delta_\kcl$.
\end{proof}

\begin{lemma}
  The expected running time of Algorithm~\ref{alg:clique-approx} is \[\mathbf{T}_{\operatorname{\textsc{Clique.Approx}}}(G, k) = O\left(\left(\frac{n}{\fkcl(G)^{1/k}} + \frac{m^{k/2}}{\fkcl(G)}\right)poly(\log{n}, 1/\alpha, k, \log{1/\delta})\right)\].
  \label{lemma:alg-approx-clique-runtime}
\end{lemma}

% \anil{change this to \cite{jain2017fast}, which is easier to see is sublinear}

\begin{proof}
  Theorem~1 of~\cite{eden2018approximating} states that the \algK~algorithm takes $O\left(\left(\frac{n}{\fkcl(G)^{1/k}} + \frac{m^{k/2}}{\fkcl(G)}\right)poly(\log{n}, 1/\alpha, k)\right)$ for the expected running time, where $\fkcl(G)$ is the number of $k$-cliques in the input graph $G$. In Algorithm~\ref{alg:clique-approx}, we invoke in total $O(\log{n} + \log{1/\delta})$ instances of the \algK~algorithm in steps $2-4$, which multiplies the running time by a factor of $O(\log{n}+\log{1/\delta})$.
\end{proof}

% \begin{theorem}
% \label{theorem:total-time-blackbox}
%   The calculation of $\tilde{S}_{\beta}(G)$ using the black-box method takes \[O\left( \left[ \sum_{u,v \in V(G)}\mathbf{T}_{\operatorname{\textsc{Clique.Approx}}}(G(A_{uv}), k-2) + m + n + \frac{||W||_1}{\hat{a}^2}\right] poly(\log{m}, \log{n}, 1/\alpha, k, \log{1/\delta})\right)\] expected running time, with $W$ as defined in Definition~\ref{def:W}.
% \end{theorem}

\noindent
%\textbf{Proof of Theorem \ref{theorem:total-time-blackbox}}

\totaltimeblackbox*

\begin{proof}
  By the definition of $\tilde{S}_{\beta}$ as above, we need to calculate $\tilde{I}_t(G)$ for up to $T = O(poly(k))$ values of $t$, which is accounted to the $poly(k)$ component of the running time.
  As being defined in Theorem~\ref{theorem:tildeS}, for each $\tilde{I}_t(G)$, we need to calculate $\widetilde{LS}(G)$, which requires $n^2$ invocations of Algorithm~\ref{alg:clique-approx} according to Definition~\ref{def:approx-ls-blackbox} with the clique size $k-2$, each for each pair of vertices $u,v$. 
  The calculation of $\tilde{a}$ (see Section~\ref{sec:fastA}) takes $O\left(m + n + \frac{||W||_1\log{m}\log{n}}{\hat{a}^2}\right)$, and the Theorem follows.
\end{proof}

\subsection{Estimation of $\ahat$}
\label{sec:fastA}

\begin{definition}
  \label{def:W}
  Let $W = \mathbb{N}^{n\times n}, W_{u,v} = deg(u) * deg(v)$ if $(u,v)\in E(G)$ and $W_{u,v} = 0$ otherwise.
\end{definition}

\begin{algorithm}[H]
  \caption{\textsc{FastA}$(G, \gamma, \delta)$\\
    %\textbf{Input:} $G(E, V), \delta, \theta < \frac{1}{2}$\\
    \textbf{Input:} $G(E, V), \gamma, \delta$\\
    \textbf{Output:} A $(\gamma, \delta)$-upper approximation of $\ahat$
  }
  \label{alg:post-processing}
  \begin{algorithmic}[1]
  \STATE $\tau_k := Algorithm~1~\cite{nguyen2023faster}(G)$
  \STATE $\tau := \tau_k / 3$
  \STATE $\theta := \frac{e^{2\gamma} -1}{e^{2\gamma}+1}$
  \STATE $c := max(\log_n(1/(2\delta)) + 2, 12\theta^2)$
 % \dungnote{May have to check $c$ to satisfies condition of Lemma 8, 9}
  \STATE Calculate $s := \frac{3c\log{n}\Vert W \Vert}{\theta^2\tau}$
  \STATE $x := Algorithm~4~\cite{nguyen2023faster}(G, s)$
  \STATE  $\tilde{a}^2_{max}:=\frac{\max_{ij} x_{ij}\frac{\Vert W \Vert}{s}}{1-\theta}$ %Set $\beta = \frac{2}{2\log(2/\delta)}$
  \STATE \textbf{Return}  $\tilde{a}_{max}$ 
  \end{algorithmic}
\end{algorithm}

\begin{lemma}
  \label{lemma:ahat-atil}
  Assume $\ahat \geq 2(k'-2)$, $k' \geq 1$ and $\ahat \leq \atil \leq e^{\gamma/(2k')}\ahat$.
  We have for every $i\geq 0$
  \begin{align*}
    \binom{\atil + i}{k' -1} \leq e^\gamma\binom{\ahat + i}{k' -1}.
  \end{align*}
\end{lemma}

\begin{proof}
  The Lemma is trivial when $k' = 1$, as the binomials on both size are $1$, and the inequality reduces to $1\leq e^\gamma$. 
  We then consider $k'\geq 2$.
  Let:
  \begin{align*}
    R_i &= \frac{\binom{\atil + i}{k' -1}}{\binom{\ahat + i}{k' -1}}
    =\prod_{j=0}^{k'-2} \frac{\atil + i - j}{\ahat + i -j}.
  \end{align*}
  We show that $R_i$ is non-increasing in $i$, i.e.:
  \begin{align*}
    \frac{R_{i+1}}{R_i} &= \frac{\atil + i + 1}{\ahat + i + 1} \times \frac{\ahat + i - (k'-2)}{\atil + i - (k'-2)}
                          \leq 1,
  \end{align*}

  which is true because $(\atil + i + 1)(\ahat + i - k' + 2) \leq (\ahat + i + 1)(\atil + i - k' + 2)$ , which is also true because $-(k'-1)\atil \leq -(k'-1)\ahat$, since $\atil \geq \ahat$.
  Since $R_i$ is non-increasing in $i$, $R_0$ takes the largest value among $R_i$.
  We have:

  \begin{align*}
    R_0 &= \prod_{j=0}^{k' -2}\frac{\atil - j}{\ahat -j}
        \leq \left(\frac{\atil - (k' - 2)}{\ahat - (k' - 2)}\right)^{k' - 1}\\
        &\leq \left(\frac{e^{\gamma/(2k')}\ahat - (k' - 2)}{\ahat - (k' - 2)}\right)^{k' - 1}\\
        &= \left(e^{\gamma/(2k')} + \frac{(k'-2)(e^{\gamma/(2k')}-1)}{\ahat-(k'-2)}  \right)^{k'-1},
  \end{align*}

  where the first inequality is because the function $\frac{\atil - j}{\ahat -j}$ is non-decreasing in $j$ as $\atil \geq \ahat$, so the greatest component is $\frac{\atil - (k'-2)}{\ahat - (k'-2)}$, and the second inequality is from the Lemma's assumption. 
  It is trivial to see that $R_0 \leq e^\gamma$ for $k' = 2$ as the second term of the upper bound of $R_0$ above is $0$.
  When $k' \geq 3$, $\ahat - (k' - 2) \geq k' - 2$, and the term $\frac{(k'-2)(e^{\gamma/(2k')}-1)}{\ahat-(k'-2)} \leq e^{\gamma/(2k')}-1$.
  Therefore $R_0 \leq (2e^{\gamma/(2k')} - 1)^{k'-1}$.
  We claim that $2e^{\gamma/(2k')} - 1 \leq e^{\gamma/(k'-1)}$, which yields $R_0 \leq e^\gamma$ and the Lemma follows.
  The rest of the proof, we show that $2e^{\gamma/(2k')} - 1 \leq e^{\gamma/(k'-1)}$.
  We have:
  \begin{align*}
    e^{\gamma/(k'-1)} - 2e^{\gamma/(2k')} + 1 &= e^{\frac{2\gamma}{2k'}\frac{k'}{k'-1}} - 2e^{\frac{\gamma}{2k'}} + 1 \\
    &\geq e^{\frac{2\gamma}{2k'}} - 2e^{\frac{\gamma}{2k'}} + 1 \\
    &= (e^{\frac{\gamma}{2k'}} - 1)^2 
    \geq 0,
  \end{align*} 
  which completes the proof.
\end{proof}

\subsection{Producing the \Turan}
\label{sec:turan-shadow}

\begin{algorithm}[H]
  \caption{\textsc{Shadow-Finder}$(G, k)$\\
    \textbf{Input:} $G(E, V), k$\\
    \textbf{Output:} A Prefixed \Turan $\Ess_k$
  }
  \label{alg:shadow-finder}
  \begin{algorithmic}[1]
    \STATE Init $\mathbf{T} = \{\emptyset, V, k\}$ and $\Ess_k = \emptyset$
    \WHILE {$\exists (P, S,\ell)\in \mathbf{T}$ such that $\rho_2(S) < 1-1/(\ell-1)$}
    \STATE Construct the degeneracy DAG $\D(G|_S)$
    \STATE Let $N^+_s$ denote the outneighborhood within $\D(G|_S)$ of $s\in S$
    \STATE Delete $(P, S, \ell)$ from $\mathbf{T}$
    \FOR {Each $s\in S$}
    \IF {$\ell \leq 2$ or $\rho_2(N^+_s) > 1 - 1/(\ell-2)$} 
    \STATE Add $(P\cup\{s\}, N^+_s, \ell-1)$ to $\Ess_k$
    \ELSE
    \STATE Add $(P\cup\{s\},N^+_s, \ell-1)$ to $\mathbf{T}$
    \ENDIF
    \ENDFOR
    \ENDWHILE
    \STATE \textbf{Return} $\Ess_k$
  \end{algorithmic}
\end{algorithm}

We present here the basic concepts of \Turan.
For the full details related to \Turan, we refer readers to~\cite{10.1145/3038912.3052636, jain2020provably}.

Let $\D(G)$ denote the degeneracy of graph $G$.

\begin{definition}
  For a graph $G$, a degeneracy ordering is a permutation of $V$ given as $v_1, \ldots, v_n$  such that for each $i\leq n$, $v_i$ is the minimum degree vertex in the subgraph induced by $v_i, \ldots, v_n$.
  The number of neighbors of $v_i$ in $v_i,\ldots, v_n$ the core number of $v_i$.
  $\D(G)$, or the degeneracy of $G$, is the largest core number in $G$.
\end{definition}

\begin{definition}
  Let $C_k(G)$ be the set of all $k$-cliques in G.
  A $k$-clique Prefixed-Shadow $\Ess_k$ for a graph $G$ is a set of triples $\{(P_i, S_i, \ell_i)\}$ where $P_i\subseteq V, S_i \subseteq V, \ell_i \in \N$ such that $\forall (P_i, S_i, \ell_i) \in \Ess_k, \forall c\in C_{\ell_i}(S_i), P_i \cup c$ is a unique $k$-clique in G and there is a bijection between $C_k(G)$ and $\cup_{(P_i, S_i, \ell_i) \in \Ess_k}\cup_{c\in C_{\ell_i}(S_i)}P_i \cup c$.
  If the multiset $\{(S_i, \ell_i)\}$ is such that $\forall(S_i, \ell_i), \rho_2(S_i)> 1-1/(\ell_i-1)$ where $\rho_2(S_i)$ represents the edge density of $S_i$, then $\Ess_k$ is a $k$-clique Prefixed-\Turan of $G$.
\end{definition}

Algorithm Shadow-Finder$(G,k)$ (Algorithm 2 of~\cite{10.1145/3038912.3052636}), that takes $G$ and $k$ as the input and returns $\Ess_k$.
Shadow-Finder$(G,k)$ takes $O(\D(G)size(\Ess_k) + m + n)$ time.

\begin{theorem}
  (Theorem 5.4 of~\cite{10.1145/3038912.3052636})
  The running time of Shadow-Finder$(G,k)$ is $O(\D(G)size(\Ess_k) + m + n)$.
  \label{theorem:shadow-runtime}
\end{theorem}

Next, algorithm \Sample~allows us to sample a clique from $\Ess_k$ with uniform probabilities.
The algorithm takes a \Turan~ as input, and for each call, returns a set of $k$ vertices.

\begin{algorithm}[H]
  \caption{\Sample$(\Ess_k)$\\
    \textbf{Input:} \Turan~$\Ess_k$\\
    \textbf{Output:} A $k-vertex$ set 
  }
  \label{alg:shadow-sample}
  \begin{algorithmic}[1]
    \STATE $w(\Ess_{k}) := \sum_{(P,S,l)\in \Ess_{k}}{|S|\choose l}$
    \STATE Set probability distribution $D$ over $\Ess_k$ such that $(P, S, \ell)\in \Ess_k$ is sampled with probability $\binom{|S|}{\ell}/w(\Ess_k)$
    \STATE Sample a $(P, S, \ell)$ from $D$
    \STATE Choose a uniform at random $\ell$-tuple $c$ from $S$
    \STATE Let $B = P \cup \{c\}$
    \STATE \textbf{Return} $B$
  \end{algorithmic}
\end{algorithm}

\begin{theorem}
\label{theorem:shadow-sample}
(Claim 4.3 of~\cite{10.1145/3038912.3052636})
The probability of any $k$-clique in $G$ being returned by a call to Sample is $1/w(\Ess_k)$.
\end{theorem}

%%% Local Variables:
%%% mode: latex
%%% TeX-master: "main"
%%% End:

\subsection{Maximum point-wise statistics estimation}
\label{sec:point-statistics}

We present here a generic framework to estimate the maximum of all points statistics, in which for each entity $i \in J$, where $J$ is the data universe, there exists a statistic $0 \le X_{i} \leq U$.
We note that $J$ can be inferred from the $\Xhat$, and not necessarily presented as an input of the Algorithms in this Section.
The goal is to estimate the maximum statistic $\Xhat = \max_{i\in J}X_{i}$ within a constant factor and with a pre-determined failure probability.
Assume that we have access to some estimator $\Xtil$ with an approximation factor $\theta$, a failure probability $\delta$, and the number of sampling iteration $s$, with the following property (point-wise accurate):

\pointwiseaccurate*

We note that $\Xtil$ can take $\theta, \delta, s$ as parameters, and returns $\Xtil_i$ for all $i\in J$, denoted as $\Xtil_{i:i\in J} = \Xtil(\theta,\delta,s)$ or $\Xtil(\theta,s)$, when the failure probability $\delta$ can be inferred from the other parameters.
An approach is to run the estimator with some iteration $s$, and takes the largest $\max_{i\in J}\Xtil_i$.
We note that the accuracy of the estimator highly depends on the number of samples $s$.
Also, the optimal value of $s$ is determined by $\Xhat$, which is also the quantity we are estimating.
Because of this, we need to guess $\Xhat$, and use the guessed value to determine $s$.
Guessing $\Xhat$ too small makes us to run many sampling iterations, reducing the efficiency of the estimator.
Guessing $\Xhat$ too high (e.g., close to $U$) helps with fewer iterations. 
However, too few iterations may not guarantee a good estimation of $X_i$.

We adapt the method by~\cite{nguyen2023faster} for estimating the maximum of all pairs dot products, generalizing it to any $(\theta,\delta)$-point-wise accurate estimator $\Xtil$.
The main idea is to start with an initial guess $\tau$ of $\Xhat$ at $U$, the maximum possible value, and use it to calculate the number of iteration $s$.
With that, we run the estimator $\Xtil$ and get the estimated point-wise statistic $\max_{i\in J}\Xtil_i$.
We then compare $\max_{i\in J}\Xtil_i$ to our guess $\tau$.
If they are close, then $\tau$ is a good guess for $\Xhat$.
If not, then we guess $\tau$ too high, then we reduce $\tau$ by a constant factor and repeat the process.
Since we start guessing $\tau = U$, we invoke at most $O(\log{U})$ instances of the estimator.
With a good guess of $\tau$, i.e., up to some constant factor of $\Xhat$, we can run the estimator one last time to find a good estimation of $\Xhat$ up to any approximation factor.

\begin{algorithm}[H]
  \caption{$\tau$ estimation\\
    \textbf{Input:} $\Xtil, U, W, \delta$\\
    %\textbf{Output:} Subset $S\subseteq V(G)$ 
    \textbf{Output:}  $\tau_t: \tau_t < \Xhat$
  }
  \label{alg:tau-estimation}
  \begin{algorithmic}[1]
    \STATE $t := 1, \tau_t := U, \theta := \frac{1}{2}, \delta' = \delta / \log_{4/3}{U}$
  \WHILE {$\tau_t \geq 1$}
    \STATE $s := \frac{W\log{2/(|J|^{-1}\delta')}}{\theta^2\tau_t}$
    \STATE $\Xtil_{i:i\in J} := \Xtil(\theta, s)$
    \IF {$\max_{i} \Xtil_{i}  < \frac{3}{2}\tau_t$}
        \STATE $\tau_{t+1} := \frac{3}{4}\tau_t $
        \STATE $t := t + 1$
    \ELSE 
        \STATE \textbf{Return} $\tau_t$
    \ENDIF
  \ENDWHILE
  \end{algorithmic}
\end{algorithm}

\begin{algorithm}[H]
  \caption{Maximum point-wise statistic estimation\\
    \textbf{Input:} $\Xtil, W, \theta \leq 1/2, \delta, \tau_t < \Xhat$\\
    \textbf{Output:}  $\max_{i\in J}\Xtil_i \in (1\pm \theta)\Xhat$
  }
  \label{alg:stat-estimation}
  \begin{algorithmic}[1]
  \STATE $\tau := \tau_t / 3$
  \STATE $s := \frac{W\log{2/(|J|^{-1}\delta)}}{\theta^2\tau}$
  \STATE $\Xtil_{i:i \in J} := \Xtil(\theta,s)$
  \STATE \textbf{Return}  $\max_{i:i\in J}\Xtil_i$ 
  
  \end{algorithmic}
\end{algorithm}

\begin{lemma}
\label{lemma:large-tau-lemma}
For any $\tau_t$ such that $\max_{i} X_{i} < \lambda \tau_t$ for some constant $0< \lambda$, there exists a constant $\lambda' > \lambda$, e.g., setting $\lambda': \frac{(\lambda'-\lambda)^2}{\lambda} \geq \theta^2$, such that with probability at least $1-\delta'$, $\max_{i\in J} \Xtil_i \leq \lambda' \tau_t$. 
\end{lemma}

\begin{proof}
  Fix a point $j\in J$, since $\Xtil$ is a $(\theta,\delta)$-point-wise accurate, with the approximation $\theta'$ and failure probability $|J|^{-1}\delta'$, we have $\Pr[\Xtil_j \notin (1\pm\theta')X_j] < \delta'$ with appropriate $s$.
There exists $\theta'$ (with respect to the fixed $j$) such that $\lambda' = \frac{(1+\theta')X_j}{\tau_t}$, it yields: $1 + \theta' = \frac{\lambda'\tau_t}{X_j}$.
With this setting, we can lower bound $\theta'$ as follows:
\begin{align*}
    \theta' & = \frac{\lambda'\tau_t}{X_j} - 1 
              = \frac{\lambda'\tau_t - X_j}{X_j} 
     \geq \frac{\lambda'\tau_t - \lambda\tau_t}{X_j}
     = \frac{\tau_t(\lambda' - \lambda)}{X_j},
\end{align*}
in which the inequality is because $X_j < \lambda\tau_t$ as stated by the Lemma.
Now, we substitute $(1 + \theta')X_j$ by $\lambda'\tau_t$:
\begin{align*}
  \Pr[\Xtil_j > \lambda'\tau_t] &\leq  \Pr[\Xtil_j > (1 + \theta')X_j]
                                \overset{(a)}{\leq} |J|^{-1}\delta',
\end{align*}

where $(a)$ can be achieved by setting $s > \frac{W\log{2/(|J|^{-1}\delta')}}{\theta'^2X_j}$ for any $j$. In fact, by setting $\lambda': \frac{(\lambda'-\lambda)^2}{\lambda} \geq \theta^2$:

\begin{align*}
\frac{W\log{2/(|J|^{-1}\delta')}}{\theta'^2X_j} &\leq \frac{W\log{2/(|J|^{-1}\delta')}}{\frac{\tau_t^2(\lambda'-\lambda)^2}{X_j^2}X_j}\\
                                                &=\frac{X_jW\log{2/(|J|^{-1}\delta')}}{\tau_t^2(\lambda'-\lambda)^2} \\
                                                &\leq \frac{\lambda W\log{2/(|J|^{-1}\delta')}}{\tau_t(\lambda'-\lambda)^2}\\
                                                &\leq \frac{W\log{2/(|J|^{-1}\delta')}}{\theta^2\tau_t}\\
                                                &= s. 
\end{align*}

Taking the union bound over all $j\in J$, with probability at least $1-\delta'$, $\max_{i\in J}\Xtil_i \leq \lambda'\tau_t$.
\end{proof}

\begin{corollary}
\label{cor:large-tau-cor}
With $\theta = \frac{1}{2}, \lambda = 1, \lambda'=\frac{3}{2}$, with probability at least $1-\delta'$, if $\max_{j}\Xtil_j > \frac{3}{2}\tau$ then $\max_{j}X_j > \tau.$
\end{corollary}

\begin{lemma}
  \label{lemma:small-tau-lemma}
For any $\tau_t$ such that $\max_{j} X_j > \lambda \tau_t$ for some constant $\lambda\geq 1$, then there exists a constant $\lambda'' \leq (1-\theta)\lambda$ such that with probability at least $1-\delta'$, $\max_j \Xtil_j \geq \lambda'' \tau_t$. 
\end{lemma}

\begin{proof}
Since $\max_j X_j > \lambda \tau_t$, there exists a point $j'$ such that $X_{j'}> \lambda\tau_t$. For such point $j'$, since the estimator $\Xtil$ is $(\theta,\delta)$-point-wise accurate, we have:

\begin{align*}
  \Pr\left[\Xtil_{j'}  < (1-\theta) X_{j'}\right] < |J|^{-1} \delta',
\end{align*}

with appropriate $s$.
Choosing some $\lambda'' \leq (1-\theta)\lambda$, we have  $\lambda'' \leq \frac{(1-\theta)X_{j'}}{\tau_t}$, or $\lambda''\tau_t \leq (1-\theta)X_{j'}$.
Substituting it to the left hand side of the above equation, we have:

\begin{align*}
  \Pr\left[\Xtil_{j'}  < \lambda''\tau_t\right] < |J|^{-1} \delta'.
\end{align*}

%NOTE: This doesn't need to invoke the union bound
Taking the union bound over all possible $j'$, we have with probability at least $1-\delta'$:

\begin{align*}
  \max_j \Xtil_{j}  \geq \lambda''\tau_t.
\end{align*}

For the rest of the proof, we argue that our choice of the number of iterations $s$ is sufficient, i.e., $s \geq \frac{W\log{2/(|J|^{-1}\delta')}}{\theta^2X_j}$ for any $j'$ such that $X_{j'} > \lambda\tau_t$.
In Algorithm~\ref{alg:tau-estimation}, we set $s = \frac{W\log{2/(|J|^{-1}\delta')}}{\theta^2\tau_t}$, hence:

\begin{align*}
  s &= \frac{W\log{2/(|J|^{-1}\delta')}}{\theta^2\tau_t}
  \geq \frac{W\log{2/(|J|^{-1}\delta')}}{\theta^2X_{j'}/\lambda}\\
  &\geq \frac{\lambda W\log{2/(|J|^{-1}\delta')}}{\theta^2X_{j'}}
  \geq \frac{W\log{2/(|J|^{-1}\delta')}}{\theta^2X_{j'}},
\end{align*}
and the Lemma follows.
%   & \leq \exp\left(-\frac{\theta^2sa_{i'j'}^2}{3\Vert W \Vert} \right)\\
%   & = \exp\left(-\frac{3\theta^2\Vert W \Vert a_{i'j'}^2c\log{n} }{3\Vert W \Vert \theta^2 \tau} \right)\\
%   &= \exp\left(-\frac{a_{i'j'}^2c\log{n} }{\tau} \right)\\
%   & \leq \exp\left(-\lambda c\log{n}\right)
%   = n^{-c\lambda}.
\end{proof}

\begin{corollary}
\label{cor:small-tau-cor}
With $\theta = \frac{1}{2}, \lambda = 3, \lambda''=\frac{3}{2}$, with probability at least $1-\delta'$, if $\max_{j}\Xtil_j < \frac{3}{2}\tau_t$ then $\max_{j}X_j< 3\tau_t.$
\end{corollary}

\begin{lemma}
  \label{lemma:binary-search-tau}
Let $\tau_t$ is the output of Algorithm~\ref{alg:tau-estimation}, with probability at least $1-\delta$, $\tau_t < \max_j X_j <4\tau_t$.
\end{lemma}

\begin{proof}
    Assume the Algorithm~\ref{alg:tau-estimation} stops at step $t$, hence outputs $\tau_t$. 
    At step $t-1$, $\max_{j}\Xtil_j < \frac{3}{2}\tau_{t-1}$, then by Corollary~\ref{cor:small-tau-cor} we have $\max_{j} X_j < 3\tau_{t-1} = 3\frac{4}{3}\tau_t = 4\tau_t$.
    At step $t$, $\max_{j}\Xtil_j > \frac{3}{2}\tau_{t}$, then by Corollary~\ref{cor:large-tau-cor} we have $\max_{j} X_j > \tau_{t}$. 
    Applying union bound on $t < \log_{4/3}U$ steps at most, with probability at least $1-\delta'\log_{4/3}U = 1-\delta$ we have $\tau_t < \max_{j}X_j <4\tau_t$, the Lemma follows.
\end{proof}

\begin{restatable}[]{lemma}{maxstat}
  \label{lemma:stat-estimation}
    For $\theta > 0, 1 > \delta > 0$, with probability at least $1-3\delta$ in  Algorithm~\ref{alg:stat-estimation},  $\max_{j}\Xtil_j\in (1\pm\theta)\max_{j}X_j$. 
\end{restatable}

\begin{proof}
  We split all points $j\in J$ into 2 sets $\{j': X_{j'} <\tau\}$ and the remaining point $\{\tilde{j}: X_{\tilde{j}} > \tau\}$.
  We will prove that the Algorithm will not output $X_{j'}$ for any point $j'$.
    
    For every $j'$ such that $X_{j'} < \tau$, applying Lemma~\ref{lemma:large-tau-lemma}, given $\theta \leq 1/2$, setting $\lambda = 1, \lambda' = \frac{3}{2}$,  we have $\max_{j'} \Xtil_{j'} < \frac{3}{2}\tau$ with probability at least $1 - \delta$ by taking the union bound on all $j'$.

    Because $\max_{j\in J} X_j > \tau_t = 3\tau$ (by the result of Lemma~\ref{lemma:binary-search-tau}), applying Lemma~\ref{lemma:small-tau-lemma}, given $\theta \leq 1/2$, setting $\lambda = 3, \lambda'' = \frac{3}{2}$, we have $\max_{j}\Xtil_j > \frac{3}{2}\tau$ with probability at least $1-\delta$ by taking the union bound on all $j$.

    By  taking the union bound on all $j$ and $j'$, with probability at least $1-2\delta$, Algorithm~\ref{alg:stat-estimation} will not output any $\Xtil_{j'}$.

    The remaining sets $\{\tilde{j}\}$ that $X_{\tilde{j}} > \tau$. Because $\Xtil$ is $(\theta,\delta)$-point-wise-accurate , $\forall \tilde{j}: \Xtil_{\tilde{j}} \in (1\pm\theta)X_{\tilde{j}}$, since the number of iterations $s=\frac{W\log{2/(|J|^{-1}\delta)}}{\theta^2\tau} > \frac{W\log{2/(|J|^{-1}\delta)}}{\theta^2X_{\tilde{j}}}$. Applying the union bound on all point $\tilde{j}$, the Lemma follows.
%    \dungnote{Check the constants}
\end{proof}

\begin{theorem}
  \label{theorem:max-stat-runtime}
  The maximum point-wise statistic estimation routine (Algorithm~\ref{alg:tau-estimation} for $\tau_t$ followed by Algorithm~\ref{alg:stat-estimation} to produce $\max_{j\in J} X_j$) takes $O\left(\frac{W\log{2/(|J|^{-1}\delta')}}{\theta^2\max_{j\in J} X_j}\right)$ iterations of the estimator $\Xtil$.
\end{theorem}

\begin{proof}
  In Algorithm~\ref{alg:tau-estimation}, the while statement loops at most $t_0$ times, with $\frac{U}{(4/3)^{t_0 - 1}} \geq \Xhat/4$ (Lemma~\ref{lemma:binary-search-tau}), or $t_0 < \log_{4/3}(4U/\Xhat)$, each invokes the estimator $\Xtil$ with $s=\frac{W\log{2/(|J|^{-1}\delta')}}{\theta^2\tau_t}$ iterations (note that $\tau_t$ depends on the iteration $t$).
  Therefore, the total number of iterations that $\Xtil$ has to run is at most $\sum_{t=1}^{t_0}\frac{W\log{2/(|J|^{-1}\delta')}}{\theta^2\tau_t} = \sum_{t=1}^{\log_{4/3}(4U/\Xhat)}\frac{W\log{2/(|J|^{-1}\delta')}(4/3)^{t-1}}{\theta^2U} = O\left(\frac{W\log{2/(|J|^{-1}\delta')}}{\theta^2\Xhat}\right)$, since $\sum_{t=1}^{\log_{4/3}{(4U/\Xhat)}} (4/3)^{t-1} = O(U/\Xhat)$ as the sum of a geometric series.
  Algorithm~\ref{alg:stat-estimation} takes $\frac{W\log{2/(|J|^{-1}\delta')}}{\theta^2\tau}$ iterations, with $\tau > \Xhat/12$ (as the result of Lemma~\ref{lemma:binary-search-tau}).
  Both algorithms require $O\left(\frac{W\log{2/(|J|^{-1}\delta')}}{\theta^2\Xhat}\right)$ iterations of $\Xtil$ in total.
\end{proof} 

%%% Local Variables:
%%% mode: latex
%%% TeX-master: "main"
%%% End:

% \input{pair_clique}
% \input{near_clique}

\end{document}